\documentclass[12pt]{article}

\usepackage[T1]{fontenc}
\usepackage{authblk}
\usepackage[titletoc]{appendix}
\usepackage[dvips]{graphicx}
\usepackage{amssymb}
\usepackage{amsmath}
\usepackage{undertilde}
\usepackage{amsfonts}
\usepackage{color}
\usepackage{amssymb}
\usepackage{ntheorem}
\usepackage{array}
\usepackage{lipsum}
\usepackage[labelsep=none]{caption}
\usepackage[mathscr]{euscript}
\newtheorem{lemma}{{\sc Lemma}}

\newtheorem{cor}{{\sc Corollary}}
\newtheorem{theorem}{{\sc Theorem}}

\newtheorem{claim}{{\sc Claim}}
\newtheorem{example}{{\sc Example}}

\newtheorem{definition}{{\sc Definition}}
\newcommand{\tbf}{\textbf}

\newcommand{\sr}{\textcolor{red}}
\newtheorem{remark}{{\sc Remark}}
\newenvironment{proof}{\noindent {\bf \sl \textbf{Proof}\/}:}
{\hfill $\blacksquare{}$ \vspace{12pt}} 

\newenvironment{proofs}{\noindent {\bf \sl Proof\/}:\enspace}
{\hfill $\square{}$ \vspace{12pt}} 

\title{\bf {\Large Characterizing Stability in Many-to-One Matching with Non-Responsive Couples}\thanks{We are especially grateful to Hans Peters, Bettina Klaus, and Ton Storcken for their careful guidance and constructive feedback.}}

\author[1]{Shashwat Khare}
\author[2]{Souvik Roy\footnote{Corresponding author: souvik.2004@gmail.com}}
\affil[1]{Independent Researcher}
\affil[2]{Statistical Sciences Division, Indian Statistical Institute, Kolkata}

\begin{document}
	\maketitle

	\begin{abstract}
		\noindent We study \emph{many-to-one} matching problems between institutions and individuals, where each institution may be matched to multiple individuals. The matching market includes \emph{couples}, who view pairs of institutions as complementary. Institutions’ preferences over sets of individuals are assumed to satisfy \emph{responsiveness}, whereas couples' preferences over pairs of institutions may \emph{violate responsiveness}.

In this setting, we first assume that all institutions share a \emph{common preference ordering} over individuals, and we establish:
\begin{enumerate}
    \item[(i)] a complete characterization of all couples' preference profiles for which a stable matching exists, under the additional assumption that couples violate responsiveness only to ensure co-location at the same institution, and
    \item[(ii)] a necessary and sufficient condition on the common institutional preference such that a stable matching exists when couples may violate responsiveness arbitrarily.
\end{enumerate}

Next, we relax the common preference assumption, requiring institutions to share a common ranking \emph{only over the members of each couple}. Under this weaker assumption, we provide:
\begin{enumerate}
    \item[(i)] a complete characterization of all couples' preferences for which a stable matching exists, and
    \item[(ii)] a sufficient condition on individuals' preferences that guarantees the existence of a stable matching.
\end{enumerate}

	\end{abstract}

	\section{Introduction}

In various practical settings, centralized matching procedures are employed to assign individuals to institutions on the opposite side of a market. Prominent examples include the matching of lawyers to legal positions in Canada, students to public schools in the United States, and doctors or senior health-care professionals to hospitals in several countries. A matching is said to be pairwise stable if there exists no institution–individual pair that can mutually deviate from the matching in a way that makes both strictly better off relative to their assigned partners.

Roth~\cite{roth1984labor} showed that it is possible to design mechanisms that incentivize only one side of the market to truthfully reveal their preferences. However, the results concerning stability have been more promising. The prevailing view is that stable matchings exist under appropriate domain restrictions. In particular, institutions must regard individuals as substitutes, and individuals must care only about the institutions to which they are matched.
	
Klaus and Klijn~\cite{klaus2005responsive} showed that a stable matching exists at every preference profile when couples' preferences satisfy \emph{responsiveness}. Responsiveness means that a couple is better off when either member is matched to a more preferred institution, keeping the other member's match fixed.\footnote{Roth and Sotomayor~\cite{roth1992two} introduced the \emph{substitutability} condition and showed that it is sufficient to ensure the existence of a stable matching. Hatfield and Milgrom~\cite{hatfield2005matching} showed that the \emph{substitutes} condition, a natural extension of substitutability, is also sufficient for stability. Later, Hatfield and Kojima~\cite{hatfield2010substitutes} showed that the \emph{bilateral substitutes} condition is sufficient for stability and that responsiveness implies bilateral substitutes.}
However, Kojima, Pathak, and Roth~\cite{kojima2010stability} observed that responsiveness is not satisfied in their dataset, as couples exhibit strong preferences to be matched to institutions located within the same geographical area.
	
Roth~\cite{roth1984marriage} was the first to point out that the existence of a stable matching is \emph{not} guaranteed at \emph{every} preference profile in the presence of couples in the labor market. This failure arises because couples may view pairs of jobs as complements, thereby violating the responsiveness condition. This motivates our aim to characterize the set of preference profiles at which a stable matching exists, even when responsiveness is not satisfied.

    We consider a specialized matching problem between a set of hospitals and a set of doctors, which includes some couples. Our focus is on the existence of stable matchings in the presence of couples. We first examine the scenario where all hospitals share a \emph{common preference} over individual doctors. This assumption is justified in contexts where hospitals rank doctors based on scores from a common examination. Hospitals' preferences over sets of doctors are then derived from this common preference using the notion of \emph{responsiveness}.

Responsiveness implies that, for any two allocations to a hospital that differ by exactly one doctor, the hospital prefers the allocation that includes the more preferred doctor (according to its individual preference ordering). Note that multiple responsive preferences over sets of doctors may correspond to a given preference over individual doctors. In our model, we allow different hospitals to adopt different responsive preferences over sets of doctors.
	
Each individual doctor is assumed to have a strict preference ordering over hospitals. The preference of a couple is derived from the individual preferences of its members. We allow couples' preferences to violate responsiveness in a controlled manner to capture their desire to be matched together. In the spirit of Dutta and Massó~\cite{dutta1997stability}, we assume that a couple prefers to be matched to the same hospital rather than to different hospitals.

We show that if all hospitals share a common preference over doctors, then a stable matching exists \emph{if and only if} each couple's preference satisfies responsiveness with respect to the more preferred member of the couple (as ranked by the hospitals' common preference).
	
Next, we consider the scenario where couples have arbitrary preferences over pairs of hospitals. This setting captures situations in which two doctors prefer to avoid being matched to the same or nearby institutions. 
We show that a stable matching exists in this scenario if and only if, for each couple, one of the following conditions holds: 
\begin{enumerate}
    \item [(i)] the members of the couple are ranked consecutively in the common preference ordering of hospitals, or
    \item [(ii)] there is at most one doctor ranked between the members, and one member of the couple is ranked at the bottom of the common preference.
\end{enumerate}

Finally, we relax the assumption that hospitals share a common preference over all individual doctors, requiring instead that they have a common preference only over the members of each couple. Under this assumption, we provide a necessary and sufficient condition on couples' preferences for the existence of a stable matching. 
Subsequently, from the hospitals’ perspective, we establish a sufficient condition on their preferences that guarantees the existence of a stable matching for any couples' preferences, even when couples may violate responsiveness to be matched together.

The remainder of the paper is organized as follows. In Section~2, we formally introduce the model. Section~3 investigates the existence of stable matchings when couples' preferences may violate responsiveness to be matched together, under the assumption that hospitals share a common preference over individual doctors. Section~4 considers the case where couples' preferences are unrestricted, while hospitals still have a common preference. Finally, Section~5 relaxes the common preference assumption for hospitals and provides sufficient conditions for the existence of a stable matching.

	\section{The framework}

We consider a many-to-one matching market between a set of doctors and a set of hospitals. Let $H$ denote the set of hospitals, and define $\bar{H} := H \cup \{\emptyset\}$, where $\emptyset$ represents the unmatched outcome for a doctor. That is, if a doctor is matched to $\emptyset$, they are effectively unmatched, i.e., not assigned to any hospital. Each hospital $h \in H$ has a finite capacity denoted by $\kappa_h \geq 2$.

We denote by $D$ the set of doctors. We assume that 
\[
D = F \cup M \cup S,
\]
where $F$, $M$, and $S$ are pairwise disjoint subsets of doctors. The doctors in $F$ are denoted by $\{f_1, \ldots, f_k\}$, and those in $M$ by $\{m_1, \ldots, m_k\}$, for some $k \in \mathbb{N}$.\footnote{By $\mathbb{N}$, we denote the set of natural numbers $\{1,2,\ldots\}$.} In particular, this means that $F$ and $M$ have the same cardinality. 
The doctors in $F$ and $M$ together form fixed couples, whereas the doctors in $S$ are not part of any couple. We refer to the doctors in $S$ as \emph{single doctors}, and those in $F$ or $M$ as \emph{non-single doctors}. The set of couples is denoted by
\[
C = \{\{f_1,m_1\}, \ldots, \{f_k,m_k\}\},
\]
and a generic couple by $c = \{f,m\}$.
	
Throughout this paper, we assume that $|H| \geq 2$, $|D| \geq 4$, and $|C| \geq 1$. That is, there are at least two hospitals and four doctors, including at least one couple. We further assume that the total number of vacancies across all hospitals equals the total number of doctors, i.e.,
\[
\sum_{h \in H} \kappa_h = |D|.
\]
\footnote{The situation where there are more (or fewer) doctors than the total capacity of hospitals can be addressed by introducing dummy hospitals (or doctors) and suitably modifying preferences over these dummies.}

	An allocation of a couple $c = \{f,m\}$ is an element $(h,h') \in \bar{H}^2$, where hospitals $h$ and $h'$ are matched to doctors $f$ and $m$, respectively. As noted earlier, one or both of $h$ and $h'$ may be $\emptyset$, indicating that the corresponding doctor(s) is(are) unmatched.

For notational convenience, we do not use braces for singleton sets.

	\subsection{Matching}
A matching is an allocation of doctors to hospitals such that the total number of doctors assigned to any hospital does not exceed its capacity, and each doctor is allocated to at most one hospital (i.e., to exactly one hospital or none). Formally, we have the following definition:

\begin{definition}
	A matching on $H \cup D$ is a mapping $\mu: H \cup D \to 2^D \cup \bar{H}$ satisfying:
	\begin{itemize}
		\item[(i)] For all $h \in H$, $\mu(h) \subseteq D$ with $|\mu(h)| \leq \kappa_h$,
		\item[(ii)] For all $d \in D$, $\mu(d) \in \bar{H}$,
		\item[(iii)] For all $d \in D$ and $h \in H$, $\mu(d) = h$ if and only if $d \in \mu(h)$.
	\end{itemize}
\end{definition}

\subsection{Preferences}
In this subsection, we introduce the notion of preferences of hospitals and doctors, along with certain restrictions on these preferences.

For a set $X$, let $\mathbb{L}(X)$ denote the set of linear orders on $X$, i.e., binary relations that are complete, reflexive, transitive, and antisymmetric. An element $R \in \mathbb{L}(X)$ is called a \emph{preference} over $X$, and we denote by $P$ its strict part. Since $R$ is antisymmetric, $x R y$ implies either $x = y$ or $x P y$.

For $P \in \mathbb{L}(X)$ and $k \leq |X|$, define $r_k(P)$ as the $k$-th ranked alternative in $P$, i.e., 
\[
r_k(P) = x \quad \text{if and only if} \quad |\{ y \in X : y R x \}| = k.
\]
Moreover, for $P \in \mathbb{L}(X)$ and $x \in X$, define the \emph{rank} of $x$ in $P$ by 
\[
r(x, P) = k \quad \text{if and only if} \quad r_k(P) = x.
\]
	
	\subsubsection{Preferences of hospitals}	
For any hospital $h \in H$, a preference of $h$ over individual doctors, denoted by $\tilde{P}_h$, is defined as an element of $\mathbb{L}(D \cup \{\emptyset\})$.
We assume that for all $d \in D$ and all $h \in H$, 
\[
d \; \tilde{P}_h \; \emptyset,
\]
meaning that a hospital always prefers having a doctor to a vacant position.

For a hospital $h$, the feasible sets of doctors, given its capacity, is defined as 
\[
\{ D' \subseteq D : |D'| \leq \kappa_h \}.
\]
A preference over feasible sets of doctors for hospital $h$ is an element of 
\[
\mathbb{L}(\{ D' \subseteq D : |D'| \leq \kappa_h \}).
\]
In what follows, we discuss how a preference of a hospital over individual doctors is extended to a preference over feasible sets of doctors by introducing the notion of \emph{responsiveness}.

Responsiveness captures the idea of separability used when extending preferences from individual elements to subsets. Informally, responsiveness states that hospitals always prefer to receive a better doctor (or set of doctors). For example, consider a preference $\tilde{P}_h$ of a hospital over the individual doctors $\{d_1, d_2, d_3, d_4\}$, where 
\[
d_1 \; \tilde{P}_h \; d_2 \; \tilde{P}_h \; d_3 \; \tilde{P}_h \; d_4.
\]
Then responsiveness implies that the pair $(d_1, d_2)$ is preferred to $(d_1, d_3)$, the set $(d_1, d_2, d_4)$ is preferred to $(d_1, d_3, d_4)$, and so forth, in the extension of $\tilde{P}_h$ to feasible sets of doctors. 

It is important to note that responsiveness does not determine how the hospital compares certain sets, for instance, the pairs $(d_1, d_4)$ and $(d_2, d_3)$. Hence, one may have multiple responsive extensions of $\tilde{P}_h$ where either pair could be preferred.

Below, we provide a formal definition of responsive extension.

\begin{definition} 
    Let $h$ be a hospital with capacity $\kappa_h$, and let $\tilde{P}_h$ be a preference of $h$ over individual doctors. A preference $P_h$ of $h$ over feasible sets of doctors satisfies \emph{responsiveness} with respect to $\tilde{P}_h$ if
    \begin{itemize}
        \item[(i)] The restriction of $P_h$ to individual doctors coincides with $\tilde{P}_h$, that is, for all $d, d' \in D \cup \{\emptyset\}$,
        \[
            d \; P_h \; d' \quad \text{if and only if} \quad d \; \tilde{P}_h \; d',
        \]
        
        \item[(ii)] For all $D' \subsetneq D$ and all $D_1, D_2 \subseteq D \setminus D'$ such that $|D' \cup D_1| \leq \kappa_h$ and $|D' \cup D_2| \leq \kappa_h$, we have
        \[
            (D' \cup D_1) \; P_h \; (D' \cup D_2) \quad \text{if and only if} \quad D_1 \; P_h \; D_2.
        \]
    \end{itemize}
\end{definition}

Next, we define the notion of common preference of hospitals over individual doctors. As the name suggests, this simply means that all hospitals share the same preference over individual doctors. Such a preference can be interpreted as a common ranking of doctors based on, for example, grades from a standardized examination. Note that hospitals may differ in how they extend this common preference to sets of feasible doctors.

\begin{definition}
    Let $\{P_h\}_{h \in H}$ be a collection of preferences of hospitals over feasible sets of doctors, and let $P^0_{hp} \in \mathbb{L}(D \cup \{\emptyset\})$. The collection $\{P_h\}_{h \in H}$ is said to satisfy \emph{Common Preference over Individual doctors (CPI)} with respect to $P^0_{hp}$ if, for all $h \in H$, the preference $P_h$ is responsive with respect to $P^0_{hp}$.
\end{definition}

Unless otherwise stated, we assume CPI for every collection of hospitals' preferences. Whenever we consider a collection of preferences satisfying CPI with respect to $P^0_{hp}$, we assume, for ease of presentation, that the indexing of doctors in couples is such that 
\[
f \; P^0_{hp} \; m
\]
for every couple $c = \{f,m\} \in C$, and that the couples in $C = \{\{f_1,m_1\}, \ldots, \{f_k,m_k\}\}$ are indexed so that
\[
m_1 \; P^0_{hp} \; m_2 \; P^0_{hp} \; \ldots \; P^0_{hp} \; m_k.
\]
This assumption is without loss of generality as we consider only one CPI at any given context.

It is worth mentioning that although the above assumption is without loss of generality, the restrictions we impose on the female member of a couple later in the paper are essentially imposed on the ``commonly preferred'' member of a couple. Thus, these restrictions do not pertain to any particular member (e.g., the female) of a couple.
	
	\subsubsection{Preferences of doctors}
Every doctor has a preference over hospitals, including the `empty' hospital $\emptyset$. Thus, a preference $P_d$ of a doctor $d \in D$ is an element of $\mathbb{L}(\bar{H})$, where $\bar{H} = H \cup \{\emptyset\}$. We assume that
\[
h \; P_d \; \emptyset \quad \text{for all } h \in H \text{ and all } d \in D,
\]
i.e., every doctor prefers being matched to some hospital rather than being unemployed.

We now proceed to define the preference of a couple based on the preferences of its individual members.

\paragraph{Preferences of couples} Each couple has a preference over the pairs of hospitals. Thus, a preference $P_c$ of a couple $c$ is an element of $\mathbb{L}(\bar{H}^2)$. Recall that an allocation $(h_1,h_2)$ for a couple $c=\{f,m\}$ means that $f$ is matched with $h_1$ and $m$ is matched with $h_2$. 

For a couple $c=\{f,m\}$ with preference $P_c$, and a hospital $h \in \bar{H}$, the conditional preference of $m$ given $h$, denoted $P_{m|h}$, is defined as the following preference of $m$: for all $h_1,h_2 \in \bar{H}$,  
\[
h_1 P_{m|h} h_2 \iff (h,h_1) P_c (h,h_2).
\]

As discussed earlier, we aim to deviate from responsiveness in a minimal way and study its implications on stability. We assume that a couple's preference is responsive except in situations where both members are matched to the same hospital. For instance, if $f$ prefers $h_1$ to $h_2$ and $m$ prefers $h_2$ to $h_1$, then—contrary to responsiveness which implies that the pair $(h_1,h_2)$ is preferred to both $(h_1,h_1)$ and $(h_2,h_2)$—we allow the couple $\{f,m\}$ to prefer either $(h_1,h_1)$ or $(h_2,h_2)$, or both, to the pair $(h_1,h_2)$. This flexibility reflects the couple’s potential benefit from staying together, which we refer to as a \emph{preference for togetherness}. We still assume that a couple prefers any allocation where both members are matched over one where at least one is unmatched.

To define this formally, we first recall the notion of responsiveness for couples’ preferences. This is structurally identical to responsiveness for hospitals' preferences, but we restate it for clarity.

\begin{definition}\label{def1} 
Let $c = \{f, m\}$ be a couple and suppose $P_f$ and $P_m$ are the individual preferences of $f$ and $m$, respectively. Then, a preference $P_c \in \mathbb{L}(\bar{H}^2)$ of the couple $c$ is called \emph{responsive} with respect to $P_f$ and $P_m$ if, for all $h, h_1, h_2 \in \bar{H}$:
\begin{enumerate}
    \item[(i)] $(h, h_1) P_c (h, h_2)$ if and only if $h_1 P_m h_2$, and
    \item[(ii)] $(h_1, h) P_c (h_2, h)$ if and only if $h_1 P_f h_2$.
\end{enumerate}
\end{definition}

We now define the notion of \emph{responsiveness violated for togetherness (RVT)}.

\begin{definition}\label{def2} 
Let $c = \{f, m\}$ be a couple, and let $P_f$ and $P_m$ be the individual preferences of $f$ and $m$, respectively. A preference $\bar{P}_c \in \mathbb{L}(\bar{H}^2)$ of $c$ satisfies \emph{responsiveness violated for togetherness (RVT)} if there exists a preference $P_c \in \mathbb{L}(\bar{H}^2)$ that is responsive with respect to $P_f$ and $P_m$, such that:
\begin{enumerate}
    \item[(i)] For all $h \in H$ and all $(h_1,h_2) \in \bar{H}^2$, if $(h,h) P_c (h_1,h_2)$, then $(h,h) \bar{P}_c (h_1,h_2)$.
    \item[(ii)] For all $(h,h'), (h_1,h_2) \in \bar{H}^2$ with $h \neq h'$ and $h_1 \neq h_2$, we have $(h,h') P_c (h_1,h_2)$ if and only if $(h,h') \bar{P}_c (h_1,h_2)$.
\end{enumerate}
\end{definition}

\noindent
Note that RVT implies that couples' preferences may violate responsiveness only in favor of being together at the same hospital. Furthermore, by taking $h_1 = h_2$ in Condition (i) of Definition~\ref{def2}, it follows that for all $h, h' \in H$, $(h,h) P_c (h',h')$ if and only if $(h,h) \bar{P}_c (h',h')$.

\subsubsection{Preference profiles and matching problems}

A \emph{preference profile} is a collection of preferences for all the doctors in $D$, all the couples in $C$, and all the hospitals in $H$, where the preferences of hospitals are assumed to be responsive. Thus, a preference profile, denoted by $\utilde{P}$, is a collection of preferences
\[
\utilde{P} = \left( \{\utilde{P}_d\}_{d \in D}, \{\utilde{P}_c\}_{c \in C}, \{\utilde{P}_h\}_{h \in H} \right),
\]
where for each $d \in D$, $\utilde{P}_d$ is a preference of doctor $d$ over $\bar{H}$, for each $c \in C$, $\utilde{P}_c$ is a preference of couple $c$ over $\bar{H}^2$, and for each $h \in H$, $\utilde{P}_h$ is a responsive preference over the feasible subsets of doctors, i.e., subsets of $D$ of size at most $\kappa_h$.

A \emph{matching problem} consists of a set of hospitals with their corresponding capacities, a set of doctors partitioned into the sets $F$, $M$, and $S$, and a preference profile.

Throughout the paper, we adhere to the following notational convention: whenever we refer to a given or fixed collection of preferences of hospitals or couples in any context, we use the superscript $0$. For instance, we use the notation $P^0_{hp}$ to denote a common preference over individual doctors (CPI), and later we will use $P^0_H$ and $P^0_C$ to denote a given collection of preferences of hospitals and couples, respectively.

\subsection{Stability}

Our model is formally equivalent to a many-to-many matching market, as each couple seeks two positions and each hospital has a capacity of at least two. Therefore, multiple notions of stability can be defined, depending on the types of permissible blocking coalitions.\footnote{See Roth~\cite{roth1984labor}, Roth~\cite{roth1984marriage}, Konishi and \"{U}nver~\cite{konishi2006groups}, and Echenique and Oviedo~\cite{echenique2006stability} for various alternative notions of stability in many-to-many matching problems.}

Blocking pairs can consist of either a hospital and a single doctor, or a pair of hospitals and a couple.

We say that a hospital $h$ is \emph{interested} in a set of doctors $D'$ at a matching $\mu$ if there exists $D'' \subseteq \mu(h)$ such that $(\mu(h) \setminus D'') \cup D' \, P_h \, \mu(h)$. In other words, a hospital is interested in a new set of doctors at a given matching if it would strictly prefer to replace a subset of its currently assigned doctors with that new set, subject to capacity constraints.

Similarly, we say that a doctor $d$ (or a couple $c$) is \emph{interested} in a hospital $h$ (or a pair $(h, h')$) at matching $\mu$ if $h \, P_d \, \mu(d)$ (respectively, $(h, h') \, P_c \, \mu(c)$). Note that if a hospital is interested in a doctor or set of doctors, or if a doctor (or couple) is interested in a hospital (or pair of hospitals), then they are not already matched together at $\mu$.

We now define the notion of an (individual) blocking pair consisting of a hospital and a single doctor.

\begin{definition}
Let $s$ be a single doctor, $h$ a hospital, and $\mu$ a matching. We say that $(h, s)$ \emph{blocks} $\mu$ if both $h$ and $s$ are interested in each other at $\mu$.
\end{definition}

That is, a hospital and a single doctor form a blocking pair if they are not matched together under $\mu$ but strictly prefer to be matched with each other instead.
	
Next, we define the notion of blocking between a pair of hospitals and a couple. A pair of hospitals and a couple, who are not already matched, block a matching if the couple prefers to be matched with that pair of hospitals, and the hospitals from that pair who are getting a new doctor from the couple are interested in getting it. Thus, the crucial thing here is that one of the members of the blocking couple might already be matched with one of the hospitals in the blocking pair. In that case, the other hospital must be interested in getting the other member of the couple. One might think that this case can be captured by our notion of (individual) blocking between the `other hospital' and the `other doctor'. Firstly, note that we have such notion of blocking only between hospitals and single doctors. Secondly, even if we define the notion of blocking between arbitrary (not necessarily single) individual doctor and hospital, that would not capture this situation as the other doctor might not be interested in the other hospital according to his/her individual preference but can be interested according to his/her couple preference.

\begin{definition}
For a couple $c=\{f,m\}$, a pair of hospitals $(h_f,h_m)$, and a matching $\mu$, we say $((h_f,h_m),c)$ blocks $\mu$ if $c$ is interested in $(h_f,h_m)$ at $\mu$, and
\begin{enumerate}
    \item[(i)] if $h_x \neq h_y$ and $\mu(x) \neq h_x$ for all $x \in \{f,m\}$, then $h_f$ is interested in $f$ and $h_m$ is interested in $m$,
    \item[(ii)] if $h_x \neq h_y$ and $h_x = \mu(x)$ and $h_y \neq \mu(y)$ for $x,y \in \{f,m\}$, then $h_y$ is interested in $y$,
    \item[(iii)] if $h_f = h_m = h$, then $h$ is interested in $\{f,m\}$.
\end{enumerate}
\end{definition}

It is worth mentioning that the blocking notion takes complementarity of a couple being accepted into account (by allowing the notion of a hospital being interested in a couple), but it does not take the couple into account when accepting single doctors and possibly removing members of a couple. In other words, there is an asymmetry here.

We consider this asymmetry in our model since it is not practical for big institutions like hospitals to consider the possibility of losing a member of a couple while removing the other member. This is because this possibility depends on factors like which hospital the removed member will join, whether the couple prefers to be together in that hospital, etc. Clearly, such situations can only be modeled by using a farsighted notion of blocking, which would complicate the model considerably.

\begin{definition}
A matching $\mu$ is \emph{stable} if it cannot be blocked.
\end{definition}

\begin{remark}
By our assumption that each hospital finds each doctor acceptable and each doctor finds each hospital acceptable, every stable matching is \emph{individually rational}.
\end{remark}

\begin{remark}\label{rem_cop}
To ease the presentation, for a couple $c=\{f,m\} \in C$ and hospitals $h_f,h_m \in H$, whenever a matching $\mu$ is blocked by $((h_f,h_m),c)$ where one of the members $x \in \{f,m\}$ of the couple was already in the corresponding hospital $h_x$ (that is, $h_x = \mu(x)$), we simply say that $\mu$ is blocked by the other pair $(h_y,y)$, where $y \neq x \in \{f,m\}$.
\end{remark}

\subsection{Two well-known algorithms for matching}

In this section, we present a well-known algorithm called the \emph{doctor proposing deferred acceptance algorithm (DPDA)}.\footnote{This directly follows from the well-known algorithm given by Gale and Shapley~\cite{gale1962college}.} However, for our purpose, we modify this algorithm slightly. We use this modified algorithm to match hospitals with doctors. In what follows, we give a short description of DPDA, where each doctor $d$ has a preference $P_d$ over hospitals and each hospital $h$ has a preference $P_h$ over feasible sets of doctors.

\medskip
\noindent \textbf{DPDA}: This algorithm has multiple stages. In stage 1, each doctor $d \in D$ proposes to his/her most preferred hospital according to $P_d$. Each hospital $h \in H$ provisionally accepts the most preferred collection of doctors according to $P_h$. If a hospital $h$ receives more than $\kappa_{h}$ proposals, then it keeps its most preferred $\kappa_h$ doctors from these proposals and rejects all others. 

Having defined stages $1,\ldots,k$, the stage $k+1$ is defined in the following way: Each unmatched (till stage $k$) doctor $d$ proposes to his/her most preferred hospital from the remaining set of hospitals who have not rejected him/her in any of the earlier stages. 

If a hospital, whose provisional list of accepted doctors is less than its capacity, receives one or more fresh proposals, then it continues to add to its accepted list (till its capacity). However, if a hospital $h$, whose provisional list of doctors is equal to its capacity, receives one or more fresh proposals from more preferred doctors, then it accepts these fresh proposals by rejecting the same number of relatively worse (according to $P_h$) doctors that it provisionally accepted earlier. 

The algorithm finally terminates when each doctor is either matched with some hospital or has been rejected by all hospitals.

\begin{remark}
In DPDA, each individual doctor proposes according to his/her individual preference. Therefore, couples do not play any role in it.
\end{remark}

Now, we present another well-known algorithm called the \emph{serial dictatorship algorithm (SDA)}. We give a short description of SDA where hospitals' preferences satisfy CPI with respect to $P^0_{hp}$. Recall that unless otherwise mentioned, we assume that hospitals' preferences satisfy the CPI property. That is, they have a common ranking, denoted by $P^0_{hp}$, over individual doctors.

\medskip
\noindent \textbf{SDA}: In SDA, the highest-ranked doctor according to $P^0_{hp}$ chooses his/her most-preferred hospital, and in general, the $j$-th ranked doctor according to $P^0_{hp}$ chooses his/her most preferred hospital among the hospitals with available vacancies after all the better (with rank less than $j$) doctors have made their choices.

\medskip
Our next remark is a standard result in matching theory.

\begin{remark}\label{rm:rm1}
Both DPDA and SDA produce the same matching when hospitals' preferences satisfy CPI.
\end{remark}
	
	\section{Non-Existence of stable matching under RVT preferences}\label{sec_3}

In this section, we explore the possibility of having a stable matching when couples' preferences satisfy the \emph{Responsiveness Violated for Togetherness (RVT)} property—that is, they are allowed to violate responsiveness in order to stay together. First, we show through an example that even when hospitals' preferences satisfy CPI, a stable matching is not guaranteed under RVT.

\begin{example}\label{exnew}\normalfont
Suppose there are two hospitals, each with capacity 2, two single doctors, and a couple (formed by two other doctors). Formally, let $H = \{h_1, h_2\}$ with $\kappa_{h_1} = \kappa_{h_2} = 2$, and let the set of doctors be $D = \{s_1, s_2, f, m\}$ where $c = \{f, m\}$ is the only couple. Consider the preference profile given in Table~\ref{tb1}. Here, both hospitals share a common preference over individual doctors, denoted by $P^0_{hp}$.
We do not present the hospitals' preferences over feasible sets of doctors because that does not play a role in this example.

\begin{table}[htbp]
\centering
\begin{tabular}{c c c c c c}
\hline
$P^0_{hp}$ & $P_{s_1}$ & $P_{s_2}$ & $P_f$ & $P_m$ & $P_c$ \\
\hline
$f$ & $h_2$ & $h_1$ & $h_2$ & $h_1$ & $(h_1,h_1)$ \\
$s_1$ & $h_1$ & $h_2$ & $h_1$ & $h_2$ & $(h_2,h_1)$ \\
$s_2$ &       &       &       &       & $(h_2,h_2)$ \\
$m$ &       &       &       &       & $(h_1,h_2)$ \\
\hline
\end{tabular}
\caption{Preference profile in Example~\ref{exnew}.}
\label{tb1}
\end{table}

Note that the couple's preference $P_c$ violates responsiveness in favor of togetherness: the pair $(h_1,h_1)$ is preferred over $(h_2,h_1)$ even though $f$ prefers $h_2$ and $m$ prefers $h_1$ individually.

We now show that no stable matching exists at this preference profile. Suppose, for the sake of contradiction, that there exists a stable matching $\mu$. Since the couple prefers being matched to some pair of hospitals over remaining unmatched, both $f$ and $m$ must be matched under $\mu$.
We now exhaustively consider all possible allocations for the couple and show that each leads to a blocking pair.

\begin{itemize}
    \item[(i)] \textbf{Suppose $\mu(c) = (h_1,h_1)$.} \\
    Since $h_1 \succ_{s_2} h_2$ and $s_2 \succ_{P^0_{hp}} m$, the pair $(h_1, s_2)$ blocks $\mu$.

    \item[(ii)] \textbf{Suppose $\mu(c) = (h_2,h_1)$.} \\
    Because $h_2 \succ_{s_1} h_1$ and $s_1 \succ_{P^0_{hp}} s_2 \succ_{P^0_{hp}} m$, it must be that $\mu(s_1) = h_2$. Moreover, by responsiveness, $\{f, m\} \succ_{h_1} \{s_2, m\}$. Combined with the fact that $(h_1,h_1) \succ_c (h_2,h_1)$, this implies that $((h_1,h_1), c)$ blocks $\mu$.

    \item[(iii)] \textbf{Suppose $\mu(c) = (h_2,h_2)$.} \\
    Since $h_2 \succ_{s_1} h_1$ and $s_1 \succ_{P^0_{hp}} m$, the pair $(h_2, s_1)$ blocks $\mu$.

    \item[(iv)] \textbf{Suppose $\mu(c) = (h_1,h_2)$.} \\
    Again, $h_2 \succ_{s_1} h_1$ and $s_1 \succ_{P^0_{hp}} s_2 \succ_{P^0_{hp}} m$ imply that $\mu(s_1) = h_2$. By responsiveness, $\{f, m\} \succ_{h_2} \{s_1, m\}$. Since $(h_2,h_2) \succ_c (h_1,h_2)$, we conclude that $((h_2,h_2), c)$ blocks $\mu$.
\end{itemize}

Since the four cases above exhaust all possible matchings for the couple $c$, it follows that no stable matching exists for the preference profile in Table~\ref{tb1}.
\end{example}
	
\section{Existence of stable matchings when couples' preferences satisfy RVT}\label{RV}

Given that the existence of stable matchings is not guaranteed when couples are allowed to violate responsiveness for togetherness, we investigate additional conditions on couples' preferences under which existence is ensured.

Let 
\[
P^0_{C} = \big(\{P^0_d\}_{d \in D \setminus S}, \{P^0_c\}_{c \in C}\big)
\]
be a given collection of preferences of non-single doctors (i.e., doctors in $D \setminus S$) and couples such that for every $c \in C$, $P^0_c$ satisfies the RVT property.

An \emph{extension} of $P^0_C$ refers to any preference profile where 
\begin{enumerate}
    \item [(i)] the preferences of non-single doctors and couples are exactly as given in $P^0_C$, and
    \item [(ii)] hospitals’ preferences satisfy the Common Preference over Individuals (CPI) property with respect to some fixed preference $P^0_{hp}$ over individual doctors.
\end{enumerate}

Recall that whenever hospitals’ preferences over feasible sets of doctors satisfy CPI with respect to $P^0_{hp}$, we assume
\[
f_i \;P^0_{hp}\; m_i \quad \text{for each couple } \{f_i,m_i\}.
\]

We now introduce a key condition called the \emph{Responsive for $F$} (RF) property, which describes a natural restriction on couples’ preferences guaranteeing the existence of stable matchings. Intuitively, RF requires that couples' preferences remain responsive with respect to the $f$-member of the couple. More precisely, if a couple moves together from a pair of hospitals $(h', h)$ to the pair $(h,h)$ (i.e., both members matched to the same hospital $h$), then the $f$-member strictly prefers hospital $h$ over $h'$. In other words, when a couple compromises to be together at a single hospital, the member $m$ always makes the concession, while $f$’s preferences remain responsive.

It is important to note that the RF property is not inherently about gender; it is simply a labeling convention that the member $f$ is the one whose preferences remain responsive, and $m$ is the member for whom responsiveness may be violated. Hence, the property ensures that responsiveness is preserved with respect to the “more preferred” member of the couple, which we denote by $f$.

\begin{definition}
A collection of preferences $P^0_C$ satisfies the \emph{Responsive for $F$} (RF) property if for every couple $c = \{f,m\} \in C$ and for all hospitals $h, h' \in H$, the following implication holds:
\[
(h,h) \;P^0_c\; (h',h) \implies h \;P^0_f\; h'.
\]
\end{definition}

We are now ready to state our main theorem which characterizes when stable matchings exist at every possible extension of $P^0_C$. Specifically, it states that a stable matching exists for every extension if and only if $P^0_C$ satisfies the RF property. For clarity, we present the two directions of the theorem separately.

\begin{theorem}\label{theo3}
\begin{enumerate}
    \item [(i)] If $P^0_C$ satisfies the RF property, then a stable matching exists at every extension of $P^0_C$.
    \item [(ii)] If $P^0_C$ does not satisfy the RF property, then there exists at least one extension of $P^0_C$ for which no stable matching exists.
\end{enumerate}
\end{theorem}

The proof of this theorem is provided in the Appendix \ref{proof1}.
		
\section{Existence of stable matchings when couples' preferences are unrestricted} \label{CP}	

In Section~\ref{RV}, we considered the case where couples may violate responsiveness solely for the purpose of being matched together at a single hospital and provided a necessary and sufficient condition for the existence of stable matchings in such settings. In this section, we go beyond the RVT restriction and allow for arbitrary violations of responsiveness in couples' preferences. In other words, we assume that a couple may have any joint preference over pairs of hospitals, regardless of the individual preferences of its members.

Note that in our model, a couple need not represent a traditional pair (e.g., a married couple). Instead, it may represent any pair of doctors with a joint preference over hospital placements. For example, we may consider two competitive or jealous individuals who prefer to be apart. This justifies our consideration of fully unrestricted couples' preferences.

In this general setting, the existence of a stable matching is no longer guaranteed. However, we show that a certain strengthening of the CPI condition on hospitals' preferences is both necessary and sufficient to ensure existence.

We introduce the notion of \emph{strong CPI (SCPI)}, which intuitively requires that the members of any couple be ranked “close enough” to each other in the common ranking $P^0_{hp}$ used by all hospitals. More precisely, SCPI requires the following:
\begin{enumerate}
    \item [(i)] If the $m$-member of a couple is not the least-ranked doctor in $D$ under $P^0_{hp}$, and if there are enough doctors to fill the capacity of at least one hospital, then the members of the couple must be ranked consecutively.
    \item [(i)] Otherwise, at most one doctor may be ranked between them.
\end{enumerate}

Formally, let us define this property. Recall that whenever hospitals’ preferences are assumed to satisfy CPI with respect to some $P^0_{hp}$, we assume $f \;P^0_{hp}\; m$ for any couple $\{f,m\}$. Also, recall that we write $r(d, P^0_{hp}) = k$ to denote that doctor $d$ has rank $k$ in $P^0_{hp}$.

\begin{definition}\label{def12}
Let hospitals' preferences satisfy CPI with respect to $P^0_{hp}$. We say that the hospitals’ preferences satisfy \emph{strong CPI (SCPI)} if for every couple $c = \{f, m\} \in C$:
\begin{enumerate}
    \item [(i)] If $r(m, P^0_{hp}) \neq |D|$, then either
    \[
    |\{ d \in D : f \;P^0_{hp}\; d \;P^0_{hp}\; m \}| = 0
    \quad \text{or} \quad
    |\{ d \in D : d \;P^0_{hp}\; m \}| < \kappa_h \text{ for all } h \in H.
    \]
    
    \item [(ii)] If $r(m, P^0_{hp}) = |D|$, then
    \[
    |\{ d \in D : f \;P^0_{hp}\; d \;P^0_{hp}\; m \}| \leq 1.
    \]
\end{enumerate}
\end{definition}

\medskip

We say a preference $P_c$ of a couple $c \in C$ is \emph{unrestricted} if it is any arbitrary element of $\mathbb{L}(\bar{H}^2)$ that satisfies the minimal requirement that the couple prefers any assignment in which both members are matched to hospitals over any assignment where at least one member is unmatched.

Now suppose that hospitals' preferences satisfy CPI with respect to some fixed preference $P^0_{hp}$ over individual doctors. We introduce the notion of \emph{extensions} of $P^0_{hp}$ to full preference profiles:
\begin{itemize}
    \item An \emph{extension} of $P^0_{hp}$ refers to any preference profile in which hospitals' preferences satisfy CPI with respect to $P^0_{hp}$.
    \item An \emph{RVT extension} of $P^0_{hp}$ is an extension where couples’ preferences satisfy RVT.
    \item An \emph{unrestricted extension} of $P^0_{hp}$ is an extension where couples’ preferences are unrestricted.
\end{itemize}

We now present our main result for this section, showing that SCPI is sufficient to guarantee the existence of a stable matching under any unrestricted extension.

\begin{theorem}\label{th2}
Let hospitals’ preferences satisfy CPI with respect to $P^0_{hp}$. If $P^0_{hp}$ satisfies SCPI, then a stable matching exists at every unrestricted extension of $P^0_{hp}$.
\end{theorem}

The proof of this theorem is provided in Appendix \ref{proof2}.

We now examine the converse of Theorem~\ref{th2}. It shows that if hospitals’ preferences satisfy CPI but \emph{violate} SCPI, then there exists at least one RVT extension of $P^0_{hp}$ in which no stable matching exists. This is a stronger claim than the direct converse, which would require unrestricted preferences. Thus, in order to witness nonexistence of a stable matching, we do not even need to consider the full generality of unrestricted couples' preferences—RVT suffices.

\begin{theorem}\label{th3}
Let hospitals’ preferences satisfy CPI with respect to $P^0_{hp}$. If $P^0_{hp}$ does not satisfy SCPI, then there exists an RVT extension of $P^0_{hp}$ at which no stable matching exists.
\end{theorem}

The proof of this theorem is provided in Appendix \ref{proof3}.

As an immediate consequence of Theorems~\ref{th2} and~\ref{th3}, we obtain the following characterization:

\begin{cor}
Let hospitals' preferences satisfy CPI with respect to $P^0_{hp}$. Then, a stable matching exists at every unrestricted extension of $P^0_{hp}$ if and only if $P^0_{hp}$ satisfies SCPI.
\end{cor}
	
\section{Matching market with non-identical hospital preferences} \label{sec_coup}

In both Section~\ref{RV} and Section~\ref{CP}, we assumed that hospitals have identical preferences over individual doctors. In this section, we relax this assumption and investigate the existence of stable matchings when hospitals may have distinct preferences.

It is evident from Example~\ref{exnew} that a stable matching may fail to exist under this generalization unless additional structural conditions are imposed. A natural candidate is the RF property on couples’ preferences. However, as we demonstrate in Example~\ref{exam2}, the RF property alone does not guarantee the existence of a stable matching in this more general setting.

Recall that a collection of couples' preferences $P^0_{C}$ satisfies the \emph{responsive for $F$ (RF)} property if for all $c = \{f, m\} \in C$ and for all $h, h' \in H$, 
\[
(h, h) \succ^0_c (h', h) \quad \Rightarrow \quad h \succ^0_f h'.
\]

\begin{example} \label{exam2} \normalfont
Consider a matching problem with hospitals $H = \{h_1, h_2, h_3\}$, each with capacity $2$, and doctors $D = \{f, m, s_1, s_2, s_3, s_4\}$, where $c = \{f, m\}$ is the only couple. The preferences are given in Table~\ref{tb}.

\begin{table}[!hbt]
    \centering
    \begin{tabular}{c c c c c c c c c c}
        \hline
        $P_{h_1}$ & $P_{h_2}$ & $P_{h_3}$ & $P_{s_1}$ & $P_{s_2}$ & $P_{s_3}$ & $P_{s_4}$ & $P_f$ & $P_m$ & $P_c$ \\
        \hline
        $s_3$ & $s_4$ & $s_3$ & $h_2$ & $h_3$ & $h_1$ & $h_2$ & $h_1$ & $h_2$ & $(h_1,h_2)$ \\
        $s_4$ & $s_3$ & $s_4$ & $h_1$ & $h_1$ & $h_2$ & $h_1$ & $h_3$ & $h_1$ & $(h_1,h_1)$ \\
        $s_1$ & $f$ & $m$ & $h_3$ & $h_2$ & $h_3$ & $h_3$ & $h_2$ & $h_3$ & $(h_1,h_3)$ \\
        $f$ & $m$ & $f$ &  &  &  &  &  &  & $(h_3,h_3)$ \\
        $m$ & $s_1$ & $s_1$ &  &  &  &  &  &  & $(h_3,h_2)$ \\
        $s_2$ & $s_2$ & $s_2$ &  &  &  &  &  &  & $(h_3,h_1)$ \\
         & & & & & & & & & $(h_2,h_2)$ \\
         & & & & & & & & & $(h_2,h_1)$ \\
         & & & & & & & & & $(h_2,h_3)$ \\
        \hline
    \end{tabular}
    \caption{}
    \label{tb}
\end{table}

We show that no stable matching exists at this preference profile. Assume, for contradiction, that $\mu$ is a stable matching.

From the preferences of $s_3$ and $h_1$, we deduce $\mu(s_3) = h_1$. Similarly, from $s_4$ and $h_2$, we get $\mu(s_4) = h_2$. Also, since $s_1 \succ_h s_2$ for all $h \in H$, it must be that $\mu(s_1) \succeq_{s_1} \mu(s_2)$. We now consider all possible allocations for the couple $c = \{f, m\}$ that are compatible with these placements and show that each leads to a blocking pair:

\begin{itemize}
    \item[$\bullet$] $\mu(c) = (h_1, h_2)$: Since $h_1 \succ_{s_1} h_3$ and $s_1 \succ_{h_1} f$, the pair $(h_1, s_1)$ blocks $\mu$.
    \item[$\bullet$] $\mu(c) = (h_1, h_3)$: Since $(h_1, h_2) \succ_c (h_1, h_3)$ and $m \succ_{h_2} s_1$, $\mu$ is blocked by $((h_1, h_2), c)$.
    \item[$\bullet$] $\mu(c) = (h_3, h_3)$: Since $(h_1, h_2) \succ_c (h_3, h_3)$ and both $m \succ_{h_2} s_1$ and $f \succ_{h_1} s_2$, $\mu$ is blocked by $((h_1, h_2), c)$.
    \item[$\bullet$] $\mu(c) = (h_3, h)$ for $h \in \{h_1, h_2\}$: Since $(h_3, h_3) \succ_c (h_3, h)$ and $m \succ_{h_3} s_2$, $\mu$ is blocked by $((h_3, h_3), c)$.
    \item[$\bullet$] $\mu(c) = (h_2, h_1)$: Since $h_1 \succ_{s_1} h_3$ and $s_1 \succ_{h_1} m$, $(h_1, s_1)$ blocks $\mu$.
    \item[$\bullet$] $\mu(c) = (h_2, h_3)$: Since $(h_3, h_3) \succ_c (h_2, h_3)$ and $f \succ_{h_3} s_2$, $\mu$ is blocked by $((h_3, h_3), c)$.
\end{itemize}

Thus, no stable matching exists at this profile.
\end{example}

\medskip

Example~\ref{exam2} shows that even when the RF property is satisfied, a stable matching may not exist under non-identical hospital preferences. This motivates considering additional structure on hospitals’ preferences.

We now introduce the \emph{common preference over couple members (CPC)} property, which stipulates that all hospitals agree on the relative ranking of the members of each couple. Following our naming convention, we assume that the $F$-member is the better-ranked member.

\begin{definition}
A collection of hospital preferences $P_H = \{P_h\}_{h \in H}$ satisfies the \emph{CPC property} if for every couple $c = \{f, m\} \in C$ and for every $h \in H$, we have $f \succ_h m$.
\end{definition}

We now show that even the combination of RF and CPC properties is not sufficient to ensure the existence of a stable matching.

\begin{example} \label{ex3} \normalfont
In this example, we exhibit a preference profile that satisfies both the CPC property and the RF property, yet admits no stable matching. Let the set of hospitals, their capacities, and the set of doctors be the same as in Example~\ref{exam2}. The preferences are given in Table~\ref{tb2}.

\begin{table}[!hbt]
    \centering
    \begin{tabular}{c c c c c c c c c c}
        \hline
        $P_{h_1}$ & $P_{h_2}$ & $P_{h_3}$ & $P_{s_1}$ & $P_{s_2}$ & $P_{s_3}$ & $P_{s_4}$ & $P_f$ & $P_m$ & $P_c$ \\
        \hline
        $s_3$ & $s_4$ & $s_3$ & $h_2$ & $h_3$ & $h_1$ & $h_2$ & $h_1$ & $h_2$ & $(h_1,h_2)$ \\
        $s_4$ & $s_3$ & $s_4$ & $h_1$ & $h_1$ & $h_2$ & $h_1$ & $h_3$ & $h_1$ & $(h_1,h_1)$ \\
        $s_1$ & $f$ & $f$ & $h_3$ & $h_2$ & $h_3$ & $h_3$ & $h_2$ & $h_3$ & $(h_1,h_3)$ \\
        $f$ & $m$ & $m$ & & & & & & & $(h_3,h_3)$ \\
        $m$ & $s_1$ & $s_1$ & & & & & & & $(h_3,h_2)$ \\
        $s_2$ & $s_2$ & $s_2$ & & & & & & & $(h_3,h_1)$ \\
        & & & & & & & & & $(h_2,h_2)$ \\
        & & & & & & & & & $(h_2,h_1)$ \\
        & & & & & & & & & $(h_2,h_3)$ \\
        \hline
    \end{tabular}
    \caption{}
    \label{tb2}
\end{table}

The proof that no stable matching exists in this profile follows the same structure as in Example~\ref{exam2} and is omitted for brevity.
\end{example}

\medskip

It follows from Example~\ref{ex3} that even when couples' preferences satisfy the RF property and hospitals’ preferences satisfy the CPC property, a stable matching may not exist. Moreover, this holds even under the additional constraint that no doctor is ranked between the two members of any couple at any hospital.

In what follows, we proceed to strengthen both the RF and CPC properties in order to obtain a sufficient condition for the existence of a stable matching.

\begin{remark}
Throughout this section, we assume that couples' preferences satisfy the RF property and hospitals' preferences satisfy the CPC property.
\end{remark}

\subsection{Condition on couples' preferences}

In view of the preceding discussion, we strengthen the RF property for couples’ preferences by introducing an additional condition, which we call the \emph{Strong Responsive for $F$} (SRF) property. We first provide an informal description of this property, followed by its formal definition.

Let $P^0_C = (\{P^0_d\}_{d \in D \setminus S}, \{P^0_c\}_{c \in C})$ denote a collection of preferences where $P^0_d$ is the preference of a non-single doctor $d$, and $P^0_c$ is the preference of couple $c$, such that each $P^0_c$ satisfies the RF property. An \emph{extension} of $P^0_C$ refers to any full preference profile in which:
\begin{enumerate}
    \item [(i)] the preferences of non-single doctors and couples are given by $P^0_C$, and
    \item [(ii)] the hospitals’ preferences satisfy the CPC property.
\end{enumerate}

We now define the SRF property. Informally, for a couple $c = \{f, m\}$, the SRF property ensures that if the couple prefers the joint placement $(h,h)$ to $(h,h')$, and if this preference requires $m$ to violate responsiveness by choosing $h$ over a more-preferred $h'$, then such a violation is justified only if $f$ ranks $h'$ higher than $h$. That is, responsiveness can be violated only at a hospital $h$ that is not top-ranked by $f$, and only in favor of a hospital that $f$ individually prefers to $h$.

\begin{definition}
A collection of preferences $P^0_C$ satisfying the RF property is said to satisfy the \emph{Strong RF (SRF)} property if for every couple $c = \{f,m\} \in C$ and all $h, h' \in H$ such that $r_1(P^0_f) \ne h$, the following implication holds:
\[
(h, h) \succ^0_c (h, h') \quad \text{and} \quad h' \succ^0_m h \quad \Rightarrow \quad h' \succ^0_f h.
\]
\end{definition}

The following lemma shows that when there are only two hospitals, the SRF and RF properties coincide.

\begin{lemma}
Suppose $|H| = 2$. Then a collection of preferences $P^0_C$ satisfies the SRF property if and only if it satisfies the RF property.
\end{lemma}

\begin{proof}
It is clear from the definition that the SRF property implies the RF property. For the converse, suppose $P^0_C$ satisfies the RF property but violates the SRF property for some couple $c = \{f,m\}$. Let $H = \{h_1, h_2\}$, and suppose, without loss of generality, that
\[
(h_1, h_1) \succ^0_c (h_1, h_2), \quad h_2 \succ^0_m h_1, \quad \text{and} \quad h_1 \succ^0_f h_2.
\]
Since $|H| = 2$, the condition $h_1 \succ^0_f h_2$ implies that $r_1(P^0_f) = h_1$, contradicting the assumption that $r_1(P^0_f) \ne h_1$. Therefore, SRF is not violated, and the lemma follows.
\end{proof}

We are now in a position to state our main result of this section, which characterizes the SRF property as a necessary and sufficient condition for the existence of a stable matching in every extension of a given collection of preferences.

\begin{theorem} \label{th_non}
Let $P^0_C$ be a collection of preferences of couples satisfying the RF property.
\begin{itemize}
    \item[(i)] If $P^0_C$ satisfies the SRF property, then a stable matching exists for every extension of $P^0_C$.
    \item[(ii)] If $P^0_C$ does not satisfy the SRF property, then there exists an extension of $P^0_C$ for which no stable matching exists.
\end{itemize}
\end{theorem}

The proof of Theorem~\ref{th_non} is deferred to the Appendix \ref{proof4}.

\subsection{Condition on hospitals' preferences}

We now strengthen the CPC property of hospitals' preferences by introducing an additional condition, which we refer to as the \emph{Strong Common Preference over Couples} (SCPC) property. We begin with an informal description of the property, followed by its formal definition.

Let $P^0_H = \{P^0_h\}_{h \in H}$ be a collection of hospitals' preferences, where each $P^0_h$ satisfies the CPC property. An \emph{extension} of $P^0_H$ refers to any full preference profile in which:
\begin{enumerate}
    \item [(i)]the preferences of hospitals are given by $P^0_H$, and
    \item [(ii)]the preferences of couples satisfy the RF property.
\end{enumerate}

The SCPC property strengthens CPC by requiring consistency across hospitals in how they compare other doctors to the more-preferred member (i.e., the $F$-member) of each couple. That is, if one hospital ranks a doctor $d$ higher than the $F$-member of some couple, then every other hospital must do the same.

\begin{definition}
A collection of hospitals’ preferences $P^0_H = \{P^0_h\}_{h \in H}$ satisfying the CPC property is said to satisfy the \emph{Strong Common Preference over Couples (SCPC)} property if, for all $h, h' \in H$, all $c = \{f, m\} \in C$, and all $d \in D$, we have:
\[
d \succ^0_h f \quad \text{if and only if} \quad d \succ^0_{h'} f.
\]
\end{definition}

Our next result shows that SCPC is a sufficient condition on hospitals’ preferences to guarantee the existence of a stable matching under any extension of $P^0_H$ where couples' preferences satisfy the RF property.

\begin{theorem}
Let $P^0_H = \{P^0_h\}_{h \in H}$ be a collection of hospitals’ preferences satisfying the SCPC property. Then a stable matching exists under every preference extension of $P^0_H$ in which couples’ preferences satisfy the RF property.
\end{theorem}

The proof of this theorem is provided in Appendix~\ref{proof5}.

\section{Conclusion}
	
In this paper, we have studied many-to-one matching problems between doctors and hospitals, where doctors may include couples with joint preferences. We first analyzed the case in which hospitals share a common preference over individual doctors. We showed that when couples are allowed to violate responsiveness only for the sake of being matched together, a stable matching exists at every preference profile if and only if the less-preferred member (according to the common hospital preference) is willing to compromise to be together with the more-preferred member. We also established necessary and sufficient conditions for the existence of a stable matching at every preference profile when couples are allowed to violate responsiveness arbitrarily.

Next, we considered the more general case where hospitals do not necessarily have identical preferences over individual doctors. We showed that, under the common preference over couple members (CPC) property, a stable matching exists at every profile if and only if the couples’ preferences satisfy the Strong Responsive for $F$ (SRF) property. Furthermore, we proved that if hospitals’ preferences satisfy the Strong CPC (SCPC) property, then a stable matching exists under every extension where couples’ preferences satisfy the RF property.

An interesting direction for future research is to explore settings where: (i) hospitals are partitioned based on geographical regions (i.e., hospitals in the same region form a partition), and (ii) couples violate responsiveness in order to be employed at hospitals within the same region. It follows from Theorem~\ref{th2} that a stable matching exists in such settings if hospitals' preferences satisfy SCPI. However, SCPI may not be necessary for the existence of a stable matching. Identifying the exact necessary and sufficient conditions in this setting remains an open problem.

    \section{Appendix: Remaining proofs}

	\subsection{Proof of Theorem \ref{theo3}}\label{proof1}

\textbf{Proof of Part $(i)$:} The proof is constructive. Suppose $P^0_{C}$ satisfies the RF property. We show that every extension $\utilde{P}$ of $P^0_C$ admits a stable matching.

Take any extension $\utilde{P}$ of $P^0_C$. Recall that by our initial assumption on CPI, we have $m_i P^0_{hp} m_j$ for all $i,j \in \{1,\ldots,k\}$ with $i < j$. We now describe an algorithm that constructs a stable matching at $\utilde{P}$.

\medskip
\noindent\textsc{Algorithm 1}: The algorithm proceeds in $k+1$ steps. Below, we describe Step 1 and a general Step $j$.

\smallskip
\noindent \textbf{Step 1}: Use the Sequential Deferred Acceptance (SDA) algorithm to match all doctors ranked above $m_1$ according to $P^0_{hp}$. Suppose $f_1$ is matched to hospital $h_1$. Then match $m_1$ using SDA where $m_1$ proposes according to the preference $P^0_{m_1|h_1}$.

\smallskip
\noindent \textbf{Step $j$} ($2 \leq j \leq k$): Having matched all doctors ranked above $m_{j-1}$ in steps $1$ through $j-1$, use SDA to match all doctors ranked between $m_{j-1}$ and $m_j$ according to $P^0_{hp}$. Suppose $f_j$ is matched to hospital $h_j$. Then match $m_j$ using SDA where $m_j$ proposes according to the preference $P^0_{m_j|h_j}$.

\smallskip
\noindent Repeat this process through Step $k$, and finally, match the remaining single doctors using SDA in Step $k+1$.

\medskip
Let $\mu$ be the matching output by Algorithm 1. We now show that $\mu$ is stable under $\utilde{P}$.

\medskip
\textbf{Blocking by a hospital-single pair:} We first show that $\mu$ cannot be blocked by a pair $(h,s)$ where $h \in H$ and $s \in S$. Suppose, for contradiction, that $(h,s)$ blocks $\mu$. By the construction of Algorithm 1, all doctors who propose before $s$ are ranked above $s$ in $P^0_{hp}$. Since $s \notin \mu(h)$, this must mean either (i) $\mu(s) \utilde{P}_s h$, or (ii) $d P^0_{hp} s$ for all $d \in \mu(h)$ and $|\mu(h)| = \kappa_h$.

If (i) holds, then $s$ does not block $\mu$ with $h$. If (ii) holds, then by responsiveness of hospital preferences, we have $\mu(h) \utilde{P}_h ((\mu(h) \setminus d) \cup s)$ for all $d \in \mu(h)$, which contradicts that $(h,s)$ is a blocking pair. Therefore, $\mu$ cannot be blocked by any hospital-single doctor pair.

\medskip
\textbf{Blocking by a hospital-couple pair:} Next, we show that $\mu$ cannot be blocked by a pair $((h_1,h_2),c)$ for some $c = \{f,m\} \in C$.

Assume for contradiction that $((h_1,h_2),c)$ blocks $\mu$. We prove in two steps that this leads to a contradiction.

\smallskip
\noindent \textbf{Step 1:} If $((h_1,h_2),c)$ blocks $\mu$, then $((\mu(f),h_2),c)$ also blocks $\mu$.

If $\mu(f) = h_1$, then the claim is trivial. Otherwise, suppose $\mu(f) \neq h_1$. We first claim that $\mu(f) P^0_f h_1$.

Suppose instead that $h_1 P^0_f \mu(f)$. Since $f$ proposes according to $P^0_f$ and all doctors before $f$ are ranked higher than $f$ in $P^0_{hp}$, it must be that $f \notin \mu(h_1)$ only if $d P^0_{hp} f$ for all $d \in \mu(h_1)$ and $|\mu(h_1)| = \kappa_{h_1}$. By responsiveness, this implies that $\mu(h_1) \utilde{P}_{h_1} ((\mu(h_1) \setminus d) \cup f)$ for all $d \in \mu(h_1)$, which contradicts the assumption that $((h_1,h_2),c)$ blocks $\mu$. Thus, we must have $\mu(f) P^0_f h_1$.

Next, we show that $(\mu(f), h_2) P^0_c (h_1, h_2)$. Suppose not; i.e., $(h_1,h_2) P^0_c (\mu(f), h_2)$. If $h_1 \neq h_2$, then the RVT (Responsiveness with Violation for Togetherness) condition implies $h_1 P^0_f \mu(f)$, a contradiction. If $h_1 = h_2$, then the RF property implies the same contradiction. Hence, $(\mu(f), h_2) P^0_c (h_1, h_2)$.

Now, since $((h_1,h_2),c)$ blocks $\mu$, it must be that for some $d \in \mu(h_2)$, we have $((\mu(h_2) \setminus d) \cup m) \utilde{P}_{h_2} \mu(h_2)$. Together with $(\mu(f), h_2) P^0_c (h_1, h_2)$, it follows that $((\mu(f), h_2), c)$ also blocks $\mu$.

\smallskip
\noindent \textbf{Step 2:} We now show that $((\mu(f), h_2), c)$ cannot block $\mu$.

Let $\mu(f) = h$. Since $(\mu(f), h_2) P^0_c (\mu(f), \mu(m))$, it follows from the definition of $P^0_{m|h}$ that $h_2 P^0_{m|h} \mu(m)$.

Since all doctors who propose before $m$ are ranked above $m$ in $P^0_{hp}$ and $m \notin \mu(h_2)$, it must be that $d P^0_{hp} m$ for all $d \in \mu(h_2)$ and $|\mu(h_2)| = \kappa_{h_2}$. By responsiveness, this implies $\mu(h_2) \utilde{P}_{h_2} ((\mu(h_2) \setminus d) \cup m)$ for all $d \in \mu(h_2)$, contradicting the assumption that $((\mu(f), h_2), c)$ blocks $\mu$.

This contradiction completes the proof that $\mu$ is stable. 

        \medskip
		
\noindent \textbf{Proof of Part $(ii)$:} Suppose $P^0_{C}$ does not satisfy the RF property. We show that there is an extension of $P^0_{C}$ with no stable matching. Since $P^0_{C}$ does not satisfy the RF property, there must exist a couple $c = \{f,m\}$ and two hospitals $h_1,h_2 \in H$ such that $(h_1,h_1) P^0_c (h_2,h_1)$ and $h_2 P^0_f h_1$. Moreover, since $h_2 P^0_f h_1$, it follows from the definition of RVT that $(h_2,h_2) P^0_c (h_1,h_2)$.

Consider a profile $\tilde{P}$ such that

\begin{enumerate}
    \item there are doctors $d_1,d_2 \in D \setminus \{f,m\}$ with 
    \[
    f P^0_{hp} d_1 P^0_{hp} d_2 P^0_{hp} m
    \]
    such that $r_1(\tilde{P}_{d_1})=h_1$ and $r_1(\tilde{P}_{d_2})=h_2$,
    
    \item we have 
    \[
    |\{ d : d P^0_{hp} f \text{ and } r_1(\tilde{P}_d) = h_2 \}| = \kappa_{h_2} - 2, \] 
  \[  |\{ d : d P^0_{hp} f \text{ and } r_1(\tilde{P}_d) = h_1 \}| = \kappa_{h_1} - 2,
    \]
    and
    \[
    |\{ d : d P^0_{hp} f \text{ and } r_1(\tilde{P}_d) = h \}| = \kappa_h \quad \text{for all } h \neq h_1, h_2,
    \]
    
    \item the preferences of all couples other than $c$ satisfy responsiveness.
\end{enumerate}

Since $\sum_{h \in H} \kappa_h = |D|$ by the construction of $\tilde{P}$, the four bottom-ranked (lowest ranked) doctors in $P^0_{hp}$ are $f, d_1, d_2, m$. We show that there is no stable matching at $\tilde{P}$. Assume for contradiction that a matching $\mu$ is stable at $\tilde{P}$. Since $\mu$ is stable at $\tilde{P}$, by the construction of $\tilde{P}$, it is straightforward that $\mu(d) = r_1(\tilde{P}_d)$ for all $d P^0_{hp} f$.

Because
\[
|\{ d : d P^0_{hp} f \text{ and } r_1(\tilde{P}_d) = h \}| = \kappa_h \quad \text{for all } h \neq h_1, h_2,
\]
stability of $\mu$ implies that the doctors $f, d_1, d_2, m$ cannot be matched to any hospital other than $h_1$ and $h_2$. Moreover, since
\[
|\{ d : d P^0_{hp} f \text{ and } r_1(\tilde{P}_d) = h_2 \}| = \kappa_{h_2} - 2,\] and 
\[|\{ d : d P^0_{hp} f \text{ and } r_1(\tilde{P}_d) = h_1 \}| = \kappa_{h_1} - 2,
\]
exactly two doctors among $f, d_1, d_2, m$ must be matched to each of $h_1$ and $h_2$.

Now, we distinguish the following cases depending on the allocation of the couple $c$ and show that $\mu$ is not stable in any of these cases.

\medskip

\begin{itemize}
    \item Suppose $\mu(c) = (h_2,h_2)$.
    
    Then, $(h_2, d_2)$ blocks $\mu$ as $r_1(\tilde{P}_{d_2}) = h_2$ and $d_2 P^0_{hp} m$.
    
    \item Suppose $\mu(c) = (h_1, h_2)$.
    
    Then, $((h_2,h_2), c)$ blocks $\mu$ as $f P^0_{hp} d_1 P^0_{hp} d_2$, and by the definition of RVT $(h_2,h_2) P^0_c (h_1,h_2)$.
    
    \item Suppose $\mu(c) = (h_1,h_1)$.
    
    Then, $(h_1, d_1)$ blocks $\mu$ as $r_1(\tilde{P}_{d_1}) = h_1$ and $d_1 P^0_{hp} m$.
    
    \item Suppose $\mu(c) = (h_2,h_1)$.
    
    Then, $((h_1,h_1), c)$ blocks $\mu$ as $f P^0_{hp} d_1 P^0_{hp} d_2$, and by the initial assumption, $(h_1,h_1) P^0_c (h_2,h_1)$.
\end{itemize}

This completes the proof.

	
\subsection{Proof of Theorem \ref{th2}}\label{proof3}
	
Suppose $P^0_{hp}$ satisfies SCPI. We show that there exists a stable matching for every unrestricted extension $\utilde{P}$ of $P^0_{hp}$. Let us partition the set of couples $C$ into two subsets $C_1$ and  $C_2$ such that for all $c=\{f,m\} \in C_1$,  $|\{d\in D: dP^0_{hp}m\}| < \kappa_h$, and for all $c=\{f,m\} \in C_2$,  $|\{d\in D: dP^0_{hp}m\}| \geq \kappa_h$. Let us index the couples in $C_2$ as  $\{f^1,m^1\},\ldots \{f^l,m^l\}$ where $m^iP^0_{hp}m^j$ for all $i, j \in \{1,\ldots,l\}$ with  $i<j$. Since $P^0_{hp}$ satisfies SCPI and $|\{d\in D: dP^0_{hp}m\}| \geq \kappa_h$  for all $c=\{f,m\} \in C_2$,  this implies  $f^iP^0_{hp}f^j$ for all $i, j \in \{1,\ldots,l\}$ with $i<j$. In the following, we present an algorithm that produces a stable matching at $\utilde{P}$. Clearly, by construction, for any $c=\{f,m\} \in C_1$ and $c' = \{f',m'\} \in C_2$, we have  $m P^0_{hp} f'$.

		\medskip
		
		\noindent \textsc{Algorithm 2}: We present the $1^{st}$ step and a general step of the algorithm.

		\noindent \textbf{Step 1}: Use SDA to match all the doctors who are ranked above $f^1$ according to $P^0_{hp}$ in the following manner. All the single doctors $s \in S$ propose according to $\utilde{P}_s$. For any couple $c = \{f,m\} \in C_1$,  $f$ proposes according to her conditional preference in $\tilde{P}_c$. More formally, $f$ first proposes to the hospital $h_f$ such that $(h_f,h_m)$ appears at the top position of $\tilde{P_c}$ for some hospital $h_m$. If $f$ is rejected by $h_f$, she proposes to the hospital  $h'_f$ that appears after $h_f$  in the $f$-component of the preference $\tilde{P}_c$. In other words,  $h_f'$ is such that there is no hospital $h_f''$ other than $h_f$ such that $(h_f'',h_m'')$ appears above $(h_f',h_m')$ for some $h_m'$ and $h_m''$ according to the preference  $\tilde{P_c}$.  $f$ continues to propose this way till she is matched. Once $f$ is matched with some hospital $h$, $m$ starts  proposing  to the hospitals according to the preference  $\tilde{P}_{m|h}$ till he is matched. Once all the doctors who are ranked above $f^1$ are matched,  $c^1$ proposes to the hospitals according to the preference $\utilde{P}_{c_1}$ till both members of $c^1$ get matched. More formally, $c^1$ first proposes to  $r_1(\utilde{P}_{c^1})$, and if any member of the  couple is rejected, then it proposes to $r_2(\utilde{P}_{c^1})$, and so on until both members of the couple are accepted by the corresponding hospitals. 
		$$\vdots$$
		\noindent \textbf{Step j}: Having matched all the doctors from the top till $m^{j-1}$  in $P^0_{hp}$ in steps $1$ to $j-1$, use SDA to match all the doctors that ranked below $m^{j-1}$ and above $r^j$ according to $P^0_{hp}$. Note that for $j>1$, there is no couple $c=\{f,m\} \in C_1$ such that $f$ or $m$ is ranked below $m^{j-1}$ and above $f^j$ according to $P^0_{hp}$. For the couple, $c^j$ proposes according to the preference $\utilde{P}_{c^j}$ till both of them are accepted by the corresponding hospitals. That is, $c^j$ first proposes to $r_1(\utilde{P}_{c^j})$, and if at least one member of the couple is rejected, then they propose to $r_2(\utilde{P}_{c^j})$, and so on, till both of them are accepted.
		$$\vdots$$
		
		\noindent Continue this process till Step $l-1$. Having matched all the doctors from the top till $m^{l-1}$ according to $P^0_{hp}$ matched  in steps $1$ to $l-1$, use SDA to match all the doctors that are ranked below $m^{l-1}$ and above $f^l$ according to $P^0_{hp}$. We distinguish the following two cases to match remaining doctors. Note here, that the remaining doctors now include $f^l$ and all the doctors ranked below $f^l$. \smallskip

		\noindent \textbf{Case 1}. Suppose there is no single doctor in between $f^l$ and $m^l$ in $P^0_{hp}$ for all $h \in H$. Let  $c^l$ propose to $r_1(\utilde{P}_{c^k})$. If at least one member of the couple is rejected, then let $c^l$ propose to $r_2(\utilde{P}_{c^l})$, and so on. Continue this process till both members of  the couple are accepted. Finally, match all the remaining doctors using SDA. \smallskip
		
		\noindent \textbf{Case 2}. Suppose there is a single doctor, say $s'$, in between $f^l$ and $m^l$ in $P^0_{hp}$. Note that by SCPI, there cannot be more than one single doctor in between $f^l$ and $m^l$. Suppose $H'$ is the set of hospitals that have at least one remaining vacancy. Let $h'$ be the worst hospital in $H'$ according to $\utilde{P}_{s'}$ and let $h \in H'$ be such that $(h,h')\utilde{R}_{c^l}(h'',h')$ for all $h'' \in H'$. Match $c^l$ with $(h,h')$ and $s'$ to the hospital that has a remaining vacancy.
		
		\medskip 
		
		Let $\mu$ be the outcome of Algorithm 2. We show that $\mu$ is stable at $\utilde{P}$.

		Assume for contradiction that $\mu$ is blocked by a hospital and a single doctor or a pair of hospitals and a couple. We complete the proof by considering the two cases of Algorithm 2 separately. 
		
		\medskip

		\noindent \textbf{Case 1.} Suppose Case 1 of Algorithm 2 holds. 
		
		First, we show that $\mu$ cannot be blocked by $(h,s)$ for some $h \in H$ and $s \in S$. Assume for contradiction that some pair $(h,s)$ blocks $\mu$. By the nature of Algorithm 2, all the doctors who propose before $s$ are ranked above $s$ according to the SCPI $P^0_{hp}$. Moreover, for any $c=\{f,m\} \in C_2$, if $f P^0_{hp}s$, then by SCPI, $m P^0_{hp} s$. Since $s \notin \mu(h)$, by the nature of Algorithm 2, we have either $\mu(s)\utilde{P}_sh$ or $d P^0_{hp} s$ for all $d \in \mu(h)$ and $|\mu(h)|=\kappa_h$. Clearly, if $\mu(s)\utilde{P}_sh$ then $s$ does not block with hospital $h$. On the other hand, if $d P^0_{hp} s$ for all $d \in \mu(h)$ and $|\mu(h)|=\kappa_h$, then by responsiveness of hospitals' preferences, we have $\mu(h)\utilde{P}_h(\mu(h)\setminus d) \cup s$ for all $d \in \mu(h)$. Therefore, hospital $h$ does not block with $s$. This contradicts that $(h,s)$ blocks $\mu$.
		
		Next we show that $\mu$ can not be blocked by $((h_1,h_2),c)$ for some $h_1,h_2 \in H$ and $c =\{f,m\}\in C$. First, we claim $c \notin C_1$. Note that if each hospital has enough vacancies to accommodate the couple together with all doctors who are ranked above it, $c$ would not get rejected by any pair of hospitals it applies to and thus, it would have been  matched to their top ranked pair of hospitals, which contradicts our assumption that  $c$ blocks $\mu$. Therefore, it must not  be the case that  $|\{d \in D: dP^0_{hp} m\} < \kappa_h$ for all $h$, which means $c \notin C_1$. 
		
		In view of the preceding claim, it follows that if $\mu$ is blocked by $((h_1,h_2),c)$, then $c \in C_2$. By the nature of Algorithm 2, couple $c$  proposes to $(h_1,h_2)$ before proposing to $(\mu(f),\mu(m))$, and some hospital, say $h_i \in \{h_1,h_2\}$, rejects the corresponding  member of the couple $c$. We distinguish the following two sub-cases.
		
		\smallskip
		
		\noindent \textbf{Case 1.1.}  Suppose $h_1 \neq h_2$. Since $h_i$ rejects a doctor from couple $c$, it must be that $h_i$ has no vacancies when $c$ proposes to $(h_1,h_2)$. Since $c$ is in $C_2$, we have that $f$ and $m$ are adjacent in $P^0_{hp}$. It follows that all the doctors in $\mu(h_i)$  are preferred to both $f$ and $m$. Therefore, $h_i$ will be worse off by removing a doctor from $\mu(h_i)$ and taking a member from the couple $c$, which  contradicts that $((h_1,h_2),c)$ blocks $\mu$.
		
		\smallskip
		
		\noindent \textbf{Case 1.2.} Suppose $h_1=h_2$. Because $h_1$ rejects at least one member of $c$, it must be that $h_1$ has less than two vacancies when $c$ proposes to $(h_1,h_1)$. Let $D'$ be the set of doctors that are present in $h_1$ at the time when $c$ makes the proposal to $(h_1,h_1)$. By SCPI, the definition of $C_2$, and the nature Algorithm 2, this implies that each doctor in $D'$ is preferred to both the doctors of the couple $c$. Again, by Algorithm 2, it follows that $D'\subseteq \mu(h_1)$. This means $h_1$ must release some doctors from $D'$ in order to block with $c$. Since hospitals' preferences over sets of individuals satisfy responsiveness, therefore, $h_1$ will be worse off by removing two doctors from $D'$ in order to take the couple. This contradicts that $((h_1,h_1),c)$ blocks $\mu$. 
		
		This completes the proof of Theorem \ref{th2} for Case 1.	
		
		\smallskip
			
		\noindent \textbf{Case 2.} Suppose Case 2 of Algorithm 2 holds. Note that after matching all the doctors from the top till $f^l$ in $P^0_{hp}$, we have exactly three vacancies left since $\sum_{h \in H} \kappa_h = |D|$. Recall from Case 2 of our Algorithm that $H'$ is the set of hospitals with at least one vacancy left, after all the doctors ranked above $f^l$ have been matched.     
		
		By similar argument as in Case $1$, (i) $\mu$ cannot be blocked by $(h,s)$ for any $s P^0_{hp} f^l$, and (ii) $\mu$ cannot be blocked by $((h_1,h_2),c)$ for any $c$ such that $c \neq c^l$. 
		
		First, we show  $\mu$ cannot be blocked by $(h,s')$, where $s'$ is the unique single doctor ranked between $f^l$ and $m^l$ in $P^0_{hp}$. Suppose not. Since $d P^0_{hp} s'$ for all $d \neq m^l$, it follows that $h \in H'$. By Algorithm 2, $\mu(s') \in H'$ and $\mu(m^l)$ is the worst hospital in $H'$ according to $\utilde{P}_{s'}$. Since $\mu$ is blocked by $(h,s')$, we have $h \utilde{P}_{s'} \mu(s') \utilde{R}_{s'} \mu(m^l)$, which implies  $h \neq \mu(m^l)$. We now show that $h \neq \mu(f^l)$. Assume for contradiction, $h=\mu(f^l)$. By our earlier argument, since $h \neq \mu(m^l)$, $\mu(f^l)=\mu(m^l)$ implies $h \neq \mu(f^l)$. Suppose $\mu(f^l) \neq \mu(m^l)$. This means all the doctors in $\mu(f^l)$ are ranked above $s'$ according to $P^0_{hp}$,  contradicting the fact that $\mu(f^l)$ and $s'$ block $\mu$. This shows $h \neq \mu(f^l)$. By the definition of Algorithm 2, $h \in \{\mu(s'),\mu(f^l),\mu(m^l)\}$. Since $h \notin \{\mu(f^l),\mu(m^l)\}$, it must be that $h=\mu(s')$, and hence $h$ and $s'$ can not block.   
		
		Now, we show that $\mu$ cannot be blocked by $((h_1,h_2),c^l)$ for some $h_1,h_2 \in H$. Since $d P^0_{hp} f^l$ for all $d \notin \{s',m^l\}$, it follows from Algorithm 2 and the definition of $H'$ that $h_1,h_2 \in H'$. We complete the proof by distinguishing the following two cases.
		
		\medskip
		
		\noindent \tbf{Case 2.1.} Suppose $h_2 = \mu(m^l)$. By Algorithm 2, $\mu(c^l) \utilde{R}_{c^l} (h_1,h_2)$ for all $h_1 \in H'$. Therefore, $c^l$ will not block with $(h_1,h_2)$.
		
		\medskip
		
		\noindent \tbf{Case 2.2.} Suppose $h_2 \neq \mu(m^l)$. By Algorithm 2, this means all the doctors in $h_2$ are preferred to $m^l$ according to $P^0_{hp}$. Therefore, $h_2$ will not block with $m^l$.
		
		This completes the proof of Theorem \ref{th2} for Case 2. 
		Since Case 1 and Case 2 are exhaustive, this completes the proof of Theorem \ref{th2}.

	\subsection{Proof of Theorem \ref{th3}}\label{proof2}
Suppose a CPI $P^0_{hp}$ does not satisfy SCPI. We show that there exists an RVT extension of $P^0_{hp}$ with no stable matching. Since $P^0_{hp}$ does not satisfy SCPI, one of the following two cases must happen:\\		
		\textbf{Case 1.}  There is  a couple $c=\{f,m\}$ such that $r(m,P^0_{hp}) \neq |D|$,  $|\{d \in D: f P^0_{hp} d P^0_{hp} m\}| > 0$ and $|\{d \in D: d P^0_{hp} m\}| \geq \kappa_h$ for some $h \in H$. Thus, there exist doctors $d_1,d_2 $ such that $fP^0_{hp} d_1 P^0_{hp} m P^0_{hp} d_2 $ and a hospital $h_1$ such that $|\{d \in D: d P^0_{hp} m\}| \geq \kappa_{h_1}$.\\
		\textbf{Case 2.} There is  a couple $c=\{f,m\}$ such that $r(m,P^0_{hp})=|D|$ and $|\{d \in D: f P^0_{hp} d P^0_{hp} m\}| > 1$. In other words, there exist doctors $d_1,d_2$ such that $fP^0_{hp}d_1P^0_{hp}d_2P^0_{hp}m$.\\ 
		
		In the following, we present an RVT extension of $P^0_{hp}$ with no stable matching for both Case 1 and Case 2. \smallskip
		
		Take hospitals $h_1,h_2\in H$ and consider a preference profile $\utilde{P}$ such that 
		\begin{enumerate}
			\item $r_1(\utilde{P}_{f})=r_1(\utilde{P}_{d_2})=h_2$ and $r_1(\utilde{P}_{m})=r_1(\utilde{P}_{d_1})=h_1$,
			\item $r_2(\utilde{P}_{f})=r_2(\utilde{P}_{d_2})=h_1$ and $r_2(\utilde{P}_{m})=r_2(\utilde{P}_{d_1})=h_2$,
			\item $(h_1,h_1)\utilde{P}_c(h_2,h_1)$ and $(h_1,h_1)\utilde{P}_c(h_2,h_2)$,
			\item $(h_1,h_2) \utilde{P}_c (h,h')$ for all $h,h' \in H$ such that $(h,h')$ does not belong to the set $\{(h_1,h_1),(h_2,h_2),(h_2,h_1),(h_1,h_2)\}$,
			\item preference $\utilde{P}_c$ satisfies responsiveness for all pairs of hospitals other than $(h_1,h_1)$,
			\item preferences of all couples other than $c$ satisfy responsiveness, 
			\item $|\{d : r_1(\utilde{P}_{d})=h_1 \mbox{ and } d P^0_{hp} m\}|=\kappa_{h_1} - 1$, and  
			\item for all $d \notin \{f,m,d_1,d_2\}$, $|\{d : r_1(\utilde{P}_{d})=h_2\}|=\kappa_{h_2} - 2$ and $|\{d : r_1(\utilde{P}_{d})=h\}|=\kappa_{h}$ for all  $h \neq h_1,h_2$. 
			
		\end{enumerate}
		
		Note that the assumption made in condition 7 is possible as $|\{d \in D: d P^0_{hp} m\}| \geq \kappa_{h_1}$ and $f P^0_{hp} m$. However, $r_1(\utilde{P}_f) \neq h_1$.   
		We show that there is no stable matching at $\utilde{P}$ for both Case 1 and Case 2. Assume for contradiction that $\mu$ is a stable matching at $\utilde{P}$. Note that by the construction of $\utilde{P}$, for all doctors $d$ such that $dP^0_{hp}f$, we must have  $\mu(d)=r_1(\utilde{P}_d)$. In the following claim, we show that $\mu(d) \in \{h_1,h_2\}$ for all $d \in \{f,m,d_1,d_2\}$.

		\begin{claim}
			For all $d \in \{f,m,d_1,d_2\}$,  $\mu(d) \in \{h_1,h_2\}$.
		\end{claim}
		\begin{proofs} First, we show $\mu(d) \in \{h_1,h_2\}$ for $d \in \{f,m\}$. Suppose $\mu(d) =h'$ for some $d \in \{f,m\}$ and some $h' \notin \{h_1,h_2\}$. We complete the proof for the case where $\mu(m)=h'$, the same for the case $\mu(f)=h'$ follows from similar arguments. Let $\mu(c)=(h,h')$ for some $h \in H$. Consider the matchings $(h,h_1)$ and $(h,h_2)$ of the couple $c$. Note that by responsiveness, $(h,h_1)\utilde{P}_c (h,h')$ and $(h,h_2)\utilde{P}_c (h,h')$. Further,  since $\sum_{h \in H} \kappa_h = |D|$ and $\mu(d)=r_1(\utilde{P}_d)$ for all doctors $d$ such that $dP^0_{hp}m$, $\mu(m)=h'$ implies that there must be a doctor $d'$ with $mP^0_{hp}d'$ such that either $d' \in \mu(h_1)$ or $d'\in \mu(h_2)$. 
			
			Suppose not, then points 7. and 8. imply that there exists a doctor $d'' \notin \{f,m,d_1,d_2\}$ such that $r_1(\utilde{P}_{d''}) \neq h_1$ and $d'' \in \mu(h_1)$ or $r_1(\utilde{P}_{d''}) \neq h_2$ and $d'' \in \mu(h_2)$. Thus, $d''$ can block $\mu$ with $r_1(\utilde{P}_d)$. If not, then by a recursive argument as above, by points 7. and 8.there exists another doctor preferred to $d''$ such that she is in not in her most preferred hospital and is matched to $r_1(\utilde{P}_d)$. Continuing like this, we get that there exists a doctor not in $\{f,m,d_1,d_2\}$ such that she is not in her most preferred hospital and can thus block $\mu$ with that hospital. 
			
			However, if there exists a doctor $d'$ with $mP^0_{hp}d'$ such that either $d' \in \mu(h_1)$ or $d'\in \mu(h_2)$ then, the couple $c$ blocks $\mu$ with either  $(h,h_1)$ or $(h,h_2)$ contradicting the stability of $\mu$. Therefore, $\mu(d) \in \{h_1,h_2\}$ for all $d \in \{f,m\}$.

			Now, we show $\mu(d) \in \{h_1,h_2\}$ for $d \in \{d_1,d_2\}$. Suppose $\mu(d) =h'$ for some $d \in \{d_1,d_2\}$ and some $h' \notin \{h_1,h_2\}$. Since $\mu(d) \in \{h_1,h_2\}$ for all $d \in \{f,m\}$ and  $\mu(d)=r_1(\utilde{P}_d)$ for all doctors $d$ such that $dP^0_{hp}f$, there must be a doctor $d'$ with $d_2 P^0_{hp} d'$  such that either $d' \in \mu(h_1)$ or $d'\in \mu(h_2)$. Because $r_k(\utilde{P}_d) \in \{h_1,h_2\}$ for all $k=1,2$, if $\mu(d')=h_1$, then $d$ blocks $\mu$ with $h_1$, and if $\mu(d')=h_2$, then $d$ blocks $\mu$ with $h_2$. This contradicts the stability of $\mu$. Therefore, $\mu(d) \in \{h_1,h_2\}$ for all $d \in \{d_1,d_2\}$.
			
			This completes the proof of Claim 1.
		\end{proofs}
		
		Now, we distinguish the following cases depending on the allocation of couple $c$ and show that $\mu$ is not stable for each of these cases. 
		
		\begin{itemize}
			\item Suppose $\mu(c) = (h_1,h_1)$. 
			
			Since $|\{d : r_1(\utilde{P}_{d})=h_1 \mbox{ and } d P^0_{hp} f\}|=\kappa_{h_1} -2$ and $f P^0_{hp} d_1$, thus $d_1 \notin \mu(h_1)$. Because $h_1\utilde{P}_{d_1}h_2$ and $d_1P^0_{hp} m$, this means $(h_1,d_1)$ blocks $\mu$.
			
			\item Suppose  $\mu(c) = (h_2,h_1)$. 
			
			Then, $((h_1,h_1),c)$ blocks $\mu$ as $fP^0_{hp}d_1P^0_{hp}d_2$ and $(h_1,h_1)\utilde{P}_c(h_2,h_1)$.

			\item Suppose  $\mu(c)=(h_2,h_2)$. 
			
			By the construction of $\utilde{P}$, $(h_1,h_1) \utilde{P}_c (h_2,h_2)$ and $h_2 \utilde{P}_{d_2} h_1$. If Case 1 holds, then $((h_1,h_1),c)$ blocks $\mu$ as $f P^0_{hp} d_1$ and $m P^0_{hp} d_2$. On the other hand, if Case 2 holds, then $(h_2,d_2)$ blocks $\mu$ as $d_2P^0_{hp}m$.
			
			\item Suppose $\mu(c) = (h_1,h_2)$. 
			
			Since $h_2 \utilde{P}_{f} h_1$, by RVT, $(h_2,h_2)\utilde{P}_c (h_1,h_2)$. This, together with the fact that $fP^0_{hp}d_1P^0_{hp}d_2$, means $\mu$ is blocked by $((h_2,h_2),c)$.
		\end{itemize}
		
		\noindent This completes the proof of Theorem \ref{th3}.

	\subsection{Proof of Theorem \ref{th_non}}\label{proof4}

\noindent	\textbf{Proof of Part $(i)$:} The proof of this part is constructive. Suppose $P^0_C$ satisfies the SRF property. We show that every extension of $\utilde{P}$ of $P^0_C$ has a stable matching.
	
	Take an extension $\utilde{P}$ of $P^0_C$. We present an algorithm which produces a stable matching at $\utilde{P}$.

	\noindent \textit{Algorithm 3}: Use DPDA, where at each stage, each single doctor $s$ proposes according to $\utilde{P}_s$,  and for any couple $c=\{f,m\}$, $f$ proposes according to $P^0_f$, and  $m$, if not already matched,  proposes according to $P^0_{m|h}$, where $h$ is the hospital $f$ proposes to. 

Let $\mu$ be the outcome of this algorithm. We make a remark that we repeatedly refer in our proof. 

\begin{remark}\label{lem5_1}
	If a doctor $d$ is rejected by some hospital $h$ at some stage of this algorithm, then $(h,d)$ cannot block $\mu$. This is because, by DPDA,  a hospital only rejects a doctor $d$ if $d' \utilde{P}_h d$ for all $d' \in \mu(h)$.
\end{remark}

First, we show that $\mu$ cannot be blocked by $(h,s)$ for some $h \in H$ and $s \in S$. Assume for contradiction that some pair $(h,s)$ blocks $\mu$. Since $s \notin \mu(h)$, this means either $\mu(s) \utilde{P}_s h$ or $s$ was rejected by $h$ at some stage of the algorithm. Clearly, if $\mu(s) \utilde{P}_s h $ then $s$ does not block with $h$. On the other hand, if $s$ had proposed to $h$ and was rejected by $h$ at an earlier stage, then by the above Lemma \ref{lem5_1}, $(h,s)$ cannot block $\mu$. 

Now, we show that $\mu$ cannot be blocked by $((h_1,h_2),c)$ for some $h_1,h_2 \in H$ and $c \in C$. Assume for contradiction that $((h_1,h_2),c)$ blocks $\mu$ for some $h_1,h_2 \in H$ and $c \in C$. We distinguish the following two cases.

\medskip

\noindent \textbf{Case 1.} Suppose $\mu(f)=r_1(P^0_f)=h_f$. 

First, we show that $\mu(m) \neq h_2$. To the contrary, suppose $\mu(m)=h_2$. Then, $(h_1,h_2) P^0_c (h_f,h_2)$ implies $h_1 P^0_f h_f$, which contradicts the fact that $h_f=r_1(P^0_f)$.

Next, we show that $(h_f,h_2)P^0_c(h_f,\mu(m))$. Since $((h_1,h_2),c)$ blocks $\mu$ and $\mu(f)=h_f$, $(h_1,h_2) P^0_c (h_f,\mu(m))$. If $h_1=h_f$, there is nothing to prove. Suppose $h_1 \neq h_f$. Then, by the responsiveness with respect to $f$, we have $(h_f,h_2)P^0_c(h_1,h_2)$. Since $(h_1,h_2)P^0_c(h_f,\mu(m))$ and $\mu(m) \neq h_2$, this implies $(h_f,h_2)P^0_c(h_f,\mu(m))$. Since $((h_1,h_2),c)$ blocks $\mu$ and $\mu(f)=h_f$, it follows that $((h_f,h_2),c)$ also blocks $\mu$. 

By the definition of $P^0_{m|h_f}$, $(h_f,h_2)P^0_c(h_f,\mu(m))$ implies $h_2 P^0_{m|h_f} \mu(m)$. Therefore, by the definition of Algorithm 3, it must be that $m$ had proposed to $h_2$ and got rejected at an earlier stage of Algorithm 3. Hence, by the definition of DPDA, $d' \utilde{P}_{h_2} m$ for all $d' \in \mu(h_2)$. Thus $((h_f,h_2),c)$ cannot block $\mu$. This completes the proof for Case 1.

\medskip

\noindent \textbf{Case 2.} Suppose $\mu(f) \neq h_f$.

We first prove the following lemma. 
\begin{lemma}\label{lem5_2}
If  $\mu(f) \neq h_f$, then	$\mu(f) R^0_f h_1$.
\end{lemma}

\begin{proofs}
	Assume for contradiction that $h_1 P^0_f \mu(f)$. Since by Algorithm 3, $f$ proposes according to $P^0_f$, this implies that $f$ had proposed to $h_1$ and got rejected. By the nature of DPDA, this implies $d \utilde{P}_{h_1} f$ for all $d \in \mu(h_1)$. Therefore, $((h_1,h_2),c)$ cannot block $\mu$, a contradiction.
\end{proofs}

\begin{lemma}\label{lem5_3}
	If $h_1=h_2=h$, then $((h,h),c)$ cannot block $\mu$.
\end{lemma}
\begin{proofs}
	First, we show that $h P^0_m \mu(m)$. Suppose not. By the SRF property, $(h,h)P^0_c(\mu(f),\mu(m))$ implies $h P^0_f \mu(f)$. Therefore, it must be that $f$ had proposed to $h$ in DPDA and got rejected, and hence $d \utilde{P}_h f$ for all $d \in \mu(h)$. This  contradicts our assumption that $((h,h),c)$ blocks $\mu$. 
	
	From Lemma \ref{lem5_2}, we know that $\mu(f) R^0_f h$. Since  the couples' preferences satisfy responsiveness with respect to $f$, we have $(\mu(f),h)R^0_c(h,h)$. Since $((h,h),c)$ blocks $\mu$, we have $(h,h)P^0_c(\mu(f),\mu(m))$. Thus, from the above argument, we have $(\mu(f),h)P^0_c(\mu(f),\mu(m))$ which implies $h P^0_{m|h} \mu(m)$. By the nature of the algorithm, this implies that $m$ had proposed to $h$ and got rejected. Thus for all $d \in \mu(h)$, $d \utilde{P}_h m$. Thus $((h,h),c)$ can not block $\mu$.	
\end{proofs}

Therefore, $h_1 \neq h_2$. Now, we know that couples only violate responsiveness for togetherness. Therefore $(h_1,h_2) P^0_c (\mu(f),\mu(m))$ and Lemma \ref{lem5_2} imply $h_2 P^0_m \mu(m)$. Since $P^0_c$ follows responsiveness for $f$, Lemma \ref{lem5_2} implies $(\mu(f),h_2)R^0_c(h_1,h_2)$. Since $((h_1,h_2),c)$ blocks $\mu$, $(h_1,h_2)P^0_c(\mu(f),\mu(m))$. This in turn means $(\mu(f),h_2)P^0_c(\mu(f),\mu(m))$. Therefore,  $h_2 P^0_{m|\mu(f)} \mu(m)$. However, this implies that  $m$ had proposed to $h_2$ and got rejected in some previous round of Algorithm 3, and hence, $d \utilde{P}_{h_2} m$ for all $d \in \mu(h_2)$. Therefore,  $((h_1,h_2),c)$ cannot block $\mu$.

Since Case 1 and Case 2 are exhaustive, this completes the proof.

\medskip

\noindent [Part (ii)] Suppose $P^0_C$ does not satisfy the SRF property. We show that there exists an extension of $P^0_C$ with no stable matching. 

Since $P^0_C$ does not satisfy the SRF property, there exists a couple $c = \{f,m\} \in C$ and hospitals $h_1,h_2, h_3 \in H$ such that $r_1(P^0_f)=h_1 \neq h_3$. Further, $(h_3,h_3)P^0_c (h_3,h_2)$ and $h_2 P^0_m h_3$ but we have $h_3 P^0_f h_2$. This yields $h_1 P^0_f h_3 P^0_f h_2$,  which in turn implies $|H| \geq 3$. Since $\kappa_h \geq 2$ for all $h \in H$, there exist at least four doctors  $\{d_1,d_2,d_3,d_4\} \in D \setminus \{f,m\}$ and let us denote the set of doctors $\{f,m,d_1,d_2,d_3,d_4\}$ by $D_1$. Consider a preference profile $\utilde{P}$ such that 
\begin{enumerate}
	\item for all $h \in H \setminus \{h_1,h_2,h_3\}$,  $|\{d:d \utilde{P}_{h} d' \mbox{ and } r_1(\utilde{P}_d)=h\}|=\kappa_{h}$ for $d' \in \{f,m,d_1,d_2,d_3,d_4\}$,
	\item  for all $h \in \{h_1,h_2,h_3\}$, $|\{d:d \utilde{P}_{h} d' \mbox{ and } r_1(\utilde{P}_d)=h\}|=\kappa_{h}-2$ for $d' \in \{f,m,d_1,d_2,d_3,d_4\}$,
	\item $P^0_{c'}$ satisfies responsiveness for all couples $c' \in C \setminus \{c\}$, and 
	\item The preferences of $h_1,h_2,h_3$ over $\{f,m,d_1,d_2,d_3,d_4\}$ and preferences of $\{f,m,d_1,d_2,d_3,d_4\}$ over $h_1,h_2,h_3$ is given by Table 3.  
\end{enumerate}

\begin{lemma}
If a matching  $\mu$ is stable at $\utilde{P}$, then  $\mu(d)=r_1(\utilde{P}_d)$ for all $d \notin \{f,m,d_1,d_2,d_3,d_4\}$.	
\end{lemma}
\begin{proofs}
		Suppose not. Then, there exists a doctor $d \notin \{f,m,d_1,d_2,d_3,d_4\}$, such that $r_1(\utilde{P}_d)=h$, but $d \notin \mu(h)$. Take any doctor $d \notin \{f,m,d_1,d_2,d_3,d_4\}$ and let $h = r_1(\utilde{P}_d)$. By the construction of the  preference profile $\utilde{P}$, we have $|\{d':d'\utilde{P}_h d \mbox{ and }r_1(\utilde{P}_{d'})=h \}|< \kappa_h$.  Therefore, there exists $d' \in \mu(h)$ and $d' \notin \{f,m,d_1,d_2,d_3,d_4\}$ with $r_1(\utilde{P}_{d'}) \neq h $ such that either $d \utilde{P}_h d'$ or $d' \utilde{P}_h d$.  Clearly,  if $d \utilde{P}_h d'$, then  $(h,d)$ blocks $\mu$, a contradiction. 
			
Suppose $d' \utilde{P}_h d$ and $r_1(\utilde{P}_{d'}) \neq h $. Let $h' = r_1(d')$. Since $r_1(d')=h'$ and $(h',d')$ does not block $\mu$ (by the assumption that $\mu$ is stable), it follows by using the same argument as in preceding paragraph that there exists a pair $(h'',d'')$ such that $h''=r_1(\utilde{P}_{d''})$,  $d'' \notin \mu(h'')$, $d'' \notin \{f,m,d_1,d_2,d_3,d_4\}$, and there exists a doctor $\hat{d} \in \mu(h'')$ such that $\hat{d} \utilde{P}_{h''} d''$. 

Continuing in this manner, by means of Part (1) and Part (2) of the definition of the preference profile $\utilde{P}$, we get hold of a pair $(h^*,d^*)$ such that $h^*=r_1(\utilde{P}_{d^*})$,  $d^* \notin \mu(h^*)$, $d^* \notin \{f,m,d_1,d_2,d_3,d_4\}$, and $d^* \utilde{P}_{h^*} \bar{d}$ for some $\bar{d} \in \mu(h^*)$. However,  this implies that $\mu$ is blocked by $(h^*,d^*)$, a contradiction. 
\end{proofs}
	
Since $\mu(d)=r_1(\utilde{P}_d)$ for all $d \notin \{f,m,d_1,d_2,d_3,d_4\}$, by the definition of the preference profile $\utilde{P}$,  we can restrict our attention to the scenario presented in Example \ref{ex3}. However, we have already argued the there is no stable matching for the scenario in Example \ref{ex3}. Therefore, $\mu$ cannot be stable.

\subsection{Proof of Theorem 5}\label{proof5}
The proof of this theorem is constructive. Suppose  $P^0_{H}$ satisfies SCPC property. We show that every extension $\utilde{P}$ of $P^0_H$ has a stable matching. 
	
	Consider  an extension $\utilde{P}$ of $P^0_H$. Since $P^0_C$ satisfies SCPC, therefore for any couple $c=\{f,m\} \in C$ and all hospitals $h,h' \in H$, $\{d: d \in D \text{ and }d P^0_h f\} = \{d:d \in D \text{ and }d d P^0_{h'} f\}$. That is the set of doctors preferred to $f$ under a hospital's preference is the same for of all the hospitals in $H$. Without loss of generality, $C=\{\{f_1,m_1\},\ldots \{f_k,m_k\}\}$ where  $f_i P^0_h f_j$ for all $i <j \in \{1,\ldots,k\}$ and all $h \in H$.  
	
	For any hospital $h \in H$, let $F^h_i=\{d: f_{i-1}P^0_h d \text{ and } d R^0_h f_i\}$ for all $i \in \{1,\ldots,k\}$. In other words, $F^h_i$ is the collection of doctors who are weakly preferred to $f_i$ and strictly less preferred to $f_{i-1}$ according to $P^0_h$. By the definition of SCPC, it follows that  $F^h_i=F^{h'}_i$ for all $h,h' \in H$. In view of this, let us denote $F^h_i$ by $F_i$, that is, let us drop the superscript $h$. 
	
	Now we  present an algorithm that produces a stable matching at $\utilde{P}$. 
	
	\medskip
	
	\noindent \textsc{Algorithm 4}: This algorithm  involves $k+1$ steps. We present the $1^{st}$ step and a general step of the algorithm. \smallskip

	\noindent \textbf{Step 1}: Use DPDA to match all the doctors in $F_1$ where all the single doctors $s \in S$ propose according to $\utilde{P}_s$ and $f_1$ proposes according to $\utilde{P}_{f_1}$. Let $f_1$ be matched to hospital $h_1$.
	$$ \vdots $$
	\noindent \textbf{Step j}:  After matching all the doctors who are ranked above and including $f_{j-1}$ in steps $1$ to $j-1$, use DPDA to match all the doctors in $F_j$, where all the singles $s \in S$ propose according to $\utilde{P}_s$ and $f_j$ proposes according to $\utilde{P}_{f_j}$. Match each (if any) $m_i \in F_j$ by using DPDA according to $\utilde{P}_{m|h_i}$, where $h_i$ is the hospital where $f_i$ is matched to.
	
	$$ \vdots $$

	\noindent Continue this process till Step $k$. Having matched all the doctors who are ranked above and including $f_k$, we proceed to match the remaining doctors in the following manner. Match each single doctor $s \in S$ using DPDA according to $\utilde{P}_s$, and, as before, match each  $m_i$ in the set of remaining doctors by using DPDA according to $\utilde{P}_{m|h_i}$, where $h_i$ is the hospital where $f_i$ is matched to. 
	
	\medskip

	Let $\mu$ be the outcome of Algorithm 4. We show that $\mu$ is stable at $\utilde{P}$. 
	
	By using similar arguments as in Lemma \ref{lem5_1}, and the fact that we use DPDA at every stage to match the hospitals, it follows that  if a doctor  $d$ is rejected by hospital $h$ at some stage of the algorithm, then $(h,d)$ can not block $\mu$. This in particular means that a hospital and a single doctor cannot block $\mu$. This is because, if a single doctor $s$ prefers a hospital $h$ to $\mu(s)$, then by the definition of DPDA,  he/she must have already proposed to the hospital $h$ and got rejected before getting matched with $\mu(s)$.

	Now we show that $((h_1,h_2),c)$ can not block $\mu$ for some $h_1,h_2 \in H$ and $c=\{f,m\} \in C$. Assume  for contradiction that $\mu$ is blocked by $((h_1,h_2),c)$.
	
	First we show that $\mu(f) \utilde{R}_f h_1$. Suppose not. Since $f$ proposes according to $\utilde{P}_f$ in Algorithm 4, it must be that $f$ had proposed to $h_1$  and got rejected at some stage of the algorithm. This implies $d P^0_{h_1} f$ for all $d \in \mu(h_1)$, and hence,  $((h_1,h_2),c)$ can not block $\mu$,  a contradiction.

	Since  $\mu(f) \utilde{R}_f h_1$ and  couples' preferences satisfy  responsiveness with respect to $f$, we have $(\mu(f),h_2)\utilde{R}_c(h_1,h_2)$. However, since $((h_1,h_2),c)$ blocks $\mu$, we have  $(h_1,h_2)\utilde{P}_c(\mu(f),\mu(m))$. Because  $\mu(f) \utilde{R}_f h_1$, this implies that $(\mu(f),h_2)\utilde{P}_c(\mu(f),\mu(m))$, and hence,  $h_2 \utilde{P}_{m|\mu(f)} \mu(m)$. Therefore, it follows that  $m$ had proposed to $h_2$ at an earlier stage of the algorithm and got rejected implying that  $d P^0_{h_2} m$ for all $d \in \mu(h_2)$. However, this contradicts that $((h_1,h_2),c)$ blocks $\mu$. 
	
	Since $\mu$ can not be blocked by single doctors or couples, $\mu$ is stable.

\bibliography{mybib}

	\bibliographystyle{amsplain}

\end{document}